\documentclass[11pt,letterpaper]{article}

\usepackage{typearea}
\typearea{15}

\usepackage{amsmath}
\usepackage{amsthm}
\usepackage{amssymb,sectsty}
\usepackage{xspace}
\usepackage{epsfig}

\usepackage{verbatim}

\usepackage{ifpdf}
\usepackage{shadow,shadethm,color}
\usepackage[compact]{titlesec}
\usepackage{algorithm, algorithmic}

\definecolor{Darkblue}{rgb}{0,0,0.4}
\definecolor{Brown}{cmyk}{0,0.81,1.,0.60}
\definecolor{Purple}{cmyk}{0.45,0.86,0,0}
\newcommand{\mydriver}{hypertex}
\ifpdf
 \renewcommand{\mydriver}{pdftex}
\fi
\usepackage[breaklinks,\mydriver]{hyperref}
\hypersetup{colorlinks=true,
            citebordercolor={.6 .6 .6},linkbordercolor={.6 .6 .6},%
citecolor=Darkblue,urlcolor=black,linkcolor=Darkblue,pagecolor=black}
\newcommand{\lref}[2][]{\hyperref[#2]{#1~\ref*{#2}}}


\newtheorem{theorem}{Theorem}[section]
\newtheorem{lemma}[theorem]{Lemma}

\newtheorem{claim}[theorem]{Claim}
\theoremstyle{definition}

\newcommand{\qedsymb}{\hfill{\rule{2mm}{2mm}}}
\renewenvironment{proof}{\begin{trivlist} \item[\hspace{\labelsep}{\bf
\noindent Proof.\/}] }{\qedsymb\end{trivlist}}%
        {\hspace*{\fill}$\Box$\par\vspace{4mm}}

\newenvironment{MyEqn}[1]{\setlength\arraycolsep{2pt}\begin{eqnarray*}
#1}{\end{eqnarray*}}%

\newcommand{\bs}[1]{{#1}}

\newcommand{\pr}[1]{{\rm Pr} \left[ #1 \right]}
\newcommand{\ex}[1]{{\rm E} \left[ #1 \right]}
\newcommand{\expar}[1]{{\rm E} [ #1 ]}

\newcommand{\opt}[1]{\mathsf{Opt(#1)}}
\newcommand{\Opt}{\mathsf{Opt}}

\newcommand{\out}{\mathrm{outliers}}
\newcommand{\profit}{\mathsf{profit}}
\newcommand{\flow}{\mathsf{flow}}
\newcommand{\free}{\mathsf{free}}
\newcommand{\charge}{\mathsf{charge}}

\sectionfont{\large} \subsectionfont{\normalsize}%

\newcommand{\MySkip}[1]{}

\newcommand{\e}{\epsilon}

\newcommand{\ts}{\textstyle}

\newcommand{\initOneLiners}{%
    \setlength{\itemsep}{0pt}
    \setlength{\parsep }{0pt}
    \setlength{\topsep }{0pt}
}
\newenvironment{OneLiners}[1][\ensuremath{\bullet}]
    {\begin{list}
        {#1}
        {\initOneLiners}}
    {\end{list}}

\begin{document}

\title{Scheduling with Outliers} \author{%
  Anupam Gupta\thanks{Computer Science Department, Carnegie Mellon
    University. Supported in part by NSF awards CCF-0448095 and
    CCF-0729022, and an Alfred P.~Sloan Fellowship.}%
  \and%
  Ravishankar Krishnaswamy$^*$%
  \and%
  Amit Kumar\thanks{Department of Computer Science \& Engineering,
    Indian Institute of Technology, Hauz Khas, New Delhi, India -
    110016.  Work partly done at MPI, Saarbr\"{u}cken, Germany.}%
  \and%
  Danny Segev
  \thanks{Sloan School of Management, Massachusetts Institute of Technology. Supported in part by NSF awards CCF-0448095 and
    CCF-0729022, and an Alfred P.~Sloan Fellowship.}
}
\begin{titlepage}
\def\thepage{}
\thispagestyle{empty}

\date{}
\maketitle

\begin{abstract}
  \medskip\noindent In classical scheduling problems, we are given jobs
  and machines, and have to schedule all the jobs to minimize some
  objective function. What if each job has a specified profit, and we
  are no longer required to process all jobs---we can schedule any
  subset of jobs whose total profit is at least a (hard) target profit
  requirement, while still approximately minimizing the objective
  function?

  \medskip\noindent We refer to this class of problems as
  \emph{scheduling with outliers}.  This model was initiated by Charikar
  and Khuller (SODA'06) on the minimum max-response time in broadcast
  scheduling. In this paper, we consider three other well-studied
  scheduling objectives: the generalized assignment problem, average
  weighted completion time, and average flow time, and provide LP-based
  approximation algorithms for them.
  Our main results are:
  \begin{itemize}
  \item For the \emph{minimum average flow time} problem on identical
    machines, we give a logarithmic approximation algorithm for the case
    of unit profits based on rounding an LP relaxation; we also show a
    matching integrality gap.  While the LP relaxation has been used
    before, the rounding algorithm is a delicate one.

  \item For the \emph{average weighted completion time} problem on
    unrelated machines, we give a constant-factor approximation. The
    algorithm is based on randomized rounding of the time-indexed LP
    relaxation strengthened by the knapsack-cover inequalities.

  \item For the \emph{generalized assignment problem} with outliers, we
    give a simple reduction to GAP without outliers to obtain an
    algorithm whose makespan is within $3$ times the optimum makespan,
    and whose cost is at most $(1+\epsilon)$ times the optimal cost.
  \end{itemize}
\end{abstract}



\medskip \medskip
\end{titlepage}


\section{Introduction}

In classical scheduling problems, we are given jobs and machines, and
have to schedule all the jobs to minimize some objective function.
\emph{What if we are given a (hard) profit constraint, and merely want
  to schedule a ``profitable'' subset of jobs?} In this paper, we
consider three widely studied scheduling objectives--- makespan,
weighted average completion time, and average flow-time---and give
approximation algorithms for these objectives in this model of
scheduling with outliers.

Formally, the \emph{scheduling with outliers} model is as follows: given
an instance of some classical scheduling problem, imagine each job $j$
also comes with a certain \emph{profit} $\pi_j$. Given a target profit
$\Pi$, the goal is now to pick a subset of jobs $S$ whose total profit $\sum_{j
  \in S} \pi_j$ is at least $\Pi$, and to schedule them to minimize the
underlying objective function. (Equivalently, we could define the
``budget'' $B = \sum_j \pi_j - \Pi$, and discard a subset of ``outlier''
jobs whose total profit is at most $B$.)
Note that this model introduces two different sources of computational
difficulty: on one hand, the task of choosing a set of jobs to achieve
the profit threshold captures the knapsack problem; on the other hand,
the underlying scheduling problem may itself be an intractable problem.

The goal of picking some subset of jobs to process as efficiently as
possible, so that we attain a minimum level of profit or ``happiness'',
is a natural one. In fact, various problems of scheduling with job
rejections have been studied previously: a common approach, studied by
Bartal et al.~\cite{BLMSS96}, has been to study ``prize-collecting''
scheduling problems (see, e.g.,~\cite{EKKSUW,BBCD03,ENW02,HSW03}), where
we attempt to minimize the scheduling objective \emph{plus the total
  profit of unscheduled jobs}. One drawback of this prize-collecting
approach is that we lose fine-grained control on the individual
quantities---the scheduling cost, and the lost profit---since we
na\"{\i}vely sum up these two essentially incomparable quantities.  In
fact, this makes our model (with a hard target constraint) interesting
also from a technical standpoint: while we can reduce the
prize-collecting problem to the target profit problem by guessing the
lost profit in the optimal prize-collecting solution, reductions in the
opposite direction are known only for a handful of problems with very
restrictive structure (see \lref[Section]{sec:related-work} for a
discussion).

To the best of our knowledge, the model we investigate was introduced by
Charikar and Khuller~\cite{CharikarK06}, who considered the problem of
minimizing the maximum response time in the context of broadcast
scheduling;
one of our results is to resolve an open problem from their paper.
Scheduling problems with outliers were also implicitly raised in the
context of model-based optimization with budgeted probes: Guha and
Munagala~\cite{GuhaM07} gave an LP-based algorithm for completion-time
scheduling with outliers which violated budgets by a constant
factor---we resolve an open problem in their paper by avoiding any
violation of the budgets.

\subsection{Our results}
\label{sec:our-results}


\medskip\noindent\textbf{GAP and makespan.} As a warm-up, we study the
Generalized Assignment Problem, a generalization of the makespan
minimization problem on unrelated machines, in \lref[Section]{sec:gap}.
For this problem, we give a simple reduction to the non-outlier version
of this problem to get a solution approximating the makespan and cost by
factors of $3$ and $(1+\e)$ respectively.  Recall that the best
non-outlier  guarantee is a
$2$-approximation~\cite{ShmoysT93} without violating the cost---however,
it is easy to show that in the presence of outliers the $(1+\e)$
loss in cost are unavoidable
unless ${\rm P} = {\rm NP}$.

\medskip\noindent\textbf{Average completion time.} We then consider the
problem of minimizing the sum of weighted completion times on unrelated
machines with release dates in \lref[Section]{sec:wsoct}.
\begin{theorem}
  \label{thm:intro-wjcj}
  For $R|r_j, \out|\sum_j w_j C_j$, there is a
  randomized $O(1)$-approximation algorithm.
\end{theorem}
Our algorithm is based on approximately solving the time-indexed LP
relaxation of Schulz and Skutella~\cite{SS99} strengthened with
\emph{knapsack-cover} inequalities followed by randomized rounding.  We
improve on this result to obtain an FPTAS for \emph{unweighted} sum of
completion times on a constant number of machines. (The best non-outlier
upper bound for $R|r_j|\sum_j w_j C_j$ is a $2$-approximation due to
Skutella~\cite{Skutella01}; this problem is also known to be
APX-hard~\cite{LST90}.)

\medskip\noindent\textbf{Average flow time.} This is the technical heart
of the paper. The problem is to minimize the average
(preemptive) flow time on identical machines $P |r_{j}, pmtn, \out | \sum
F_j$. Our main result is:
\begin{theorem}
  \label{thm:intro-flow}
  For $P |r_{j}, pmtn, \out |\sum F_j$, when all jobs have
  unit profits, there is an $O(\log P)$-approximation algorithm, where
  $P$ is the ratio between the largest and smallest processing times.
\end{theorem}
This comes close to matching the best known result of $O(\log \min\{P,
n/m\})$ for the non-outlier version due to Leonardi and Raz~\cite{LeoR}.
However, this problem seems to be much harder with outliers, as we get
the same approximation even on a single machine, in contrast to the
non-outlier single-machine case (which can be solved optimally). We
show our approach is tight, as the LP relaxation we use has an $\Omega( \log P)$
integrality gap.


The algorithm rounds a linear-programming relaxation originally
suggested in~\cite{GargK06}; however, we need new ideas for the rounding
algorithm over those used by~\cite{GargK06}. At a high-level, here is
the idea behind our rounding algorithm: the LP might have scheduled each
job to a certain fractional amount, and hence we try to swap ``mass''
between jobs of near-equal processing times in order to integrally
schedule a profitable subset of jobs. However, this swapping operation
is a delicate one, and merely swapping mass locally between nearby jobs
has a bad algorithmic gap.  Furthermore, we need to handle jobs that are
only approximately equal in size, which leads to additional
difficulties. (For a more detailed high-level sketch of these issues,
please read \lref[Section]{sec:plan}.)





\subsection{Related work}
\label{sec:related-work}

\noindent\textbf{Scheduling with rejections.}
As mentioned above, previous papers on this topic considered the
``prize-collecting'' version which minimizes the scheduling objective
plus the total profit of unscheduled jobs; their techniques do not seem
to extend to scheduling with outliers, in which we have a strict budget
on the total penalty of rejected jobs. Bartal et al.~\cite{BLMSS96}
considered offline and online makespan minimization and gave
best-possible algorithms for both cases. Makespan minimization with
preemptions was investigated by~\cite{HSW03,Seiden01}. Epstein et
al.~\cite{ENW02} examined scheduling unit-length jobs. Engels et
al.~\cite{EKKSUW} studied the prize-collecting version of weighted
completion-time minimization
(on single or parallel machines), and gave PTASs or constant-factor
approximations for these problems; they also proposed a general
framework for designing algorithms for such problems.


\medskip\noindent\textbf{Outlier versions of other problems.} Also
called {\em partial-covering} problems, these have been widely studied:
e.g., the $k$-MST problem~\cite{Garg05}, the $k$-center and facility
location problem~\cite{CKMN01} and the $k$-median problem with
outliers~\cite{Chen08}, partial vertex cover (e.g.,~\cite{Mestre05} and
references therein) and $k$-multicut~\cite{GNS06,LevinSegev06}. Chudak
et al.~\cite{CRW04} distilled ideas of Jain and Vazirani~\cite{JV01} on
converting ``Lagrange-multiplier preserving'' algorithms for
prize-collecting Steiner tree into one for $k$-MST; K{\"o}nemann et
al.~\cite{KPS06} gave a general framework to convert prize-collecting
algorithms into algorithms for outlier versions (see
also~\cite{Mestre08}).  We cannot use these results, since it is not
clear how to make the algorithms for prize-collecting scheduling
problems to also be Langrange-multiplier preserving, or whether the
above-mentioned framework is applicable in scheduling-related scenarios.




\section{GAP and Makespan} \label{sec:gap}

As a warm-up, we consider the generalized assignment problem, which is an
extension of minimizing makespan on unrelated machines with outliers.
Formally, the instance ${\cal I}$ has $m$ machines and $n$ jobs. Each
job $j$ has a processing time of $p_{ij}$ on machine $i$, an assignment
cost of $c_{ij}$, and a profit of $\pi_j$. Given a profit requirement
$\Pi$, cost bound $C$ and makespan bound $T$, the goal is to obtain a
feasible schedule satisfying these requirements (or to declare
infeasibility). Of course, since the problem is NP-hard, we look at finding solutions where we
violate the cost and makespan bounds, but not the (hard) profit
requirement. We now show how to reduce this problem to the non-outlier
version studied earlier, while incurring small additional losses in the
approximation guarantees.

\begin{theorem}
  \label{thm:gap-main}
  Given an instance $\cal I$ of GAP-with-outliers with optimal cost $C$,
  and makespan $T$, there is a polynomial time algorithm to output an
  assignment with cost $(1+\epsilon)C$ and makespan $3T$.
\end{theorem}
\begin{proof}
  Given the instance $\cal I$, construct the following instance ${\cal
    I}'$ of the standard GAP (where there are no profits or outliers).
  There are $m+1$ machines: machines $1,2,\ldots, m$ are the same as
  those in ${\cal I}$, while machine $m+1$ is a ``virtual profit
  machine''. We have $n$ jobs, where job $j$ has a processing time of
  $p_{ij}$ and an assignment cost of $c_{ij}$ when scheduled on machine
  $i$ (for $1 \leq i \leq m$). If job $j$ is scheduled on the virtual
  machine $m+1$, it incurs a processing time of $\pi_j$ and cost zero:
  i.e., $p_{(m+1) j} = \pi_j$ and $c_{(m+1) j} = 0$. For this instance
  $\cal I'$, we set a cost bound of $C$, makespan bound of $T$ for all
  machines $1 \leq i \leq m$, and a makespan bound of $T_{vpm} :=
  (\sum_{j=1}^{n} \pi_j) - \Pi$ for the virtual profit machine. Note
  that any feasible solution for ${\cal I}$ is also feasible for ${\cal
    I}'$, with the outliers being scheduled on the virtual profit
  machine, since the total profit of the outliers is at most $T_{vpm} =
  (\sum_{j=1}^{n} \pi_j) - \Pi$.

  We can now use the algorithm of Shmoys and Tardos~\cite{ShmoysT93}
  which guarantees an assignment $\mathbf{S}$ for the GAP instance
  ${\cal I}'$ with the following properties: \emph{(a)} The cost of
  assignment $\mathbf{S}$ is at most $C$, \emph{(b)} the makespan
  induced by $\mathbf{S}$ on machine $i$ (for $1 \leq i \leq m$) is at
  most $T$ + $\min\{\max_{j} p_{ij}, T\}$, and \emph{(c)} The makespan
  of $\mathbf{S}$ on the virtual machine $m+1$ is at most
  $((\sum_{i=1}^{n} \pi_j) - \Pi) + \max_{j} \pi_{j}$.

  Note that this assignment ${\bf S}$ is \emph{almost} feasible for the
  outlier problem ${\cal I}$---the makespan on any real machine is at
  most $T + \max_{j} p_{ij}$, the assignment cost is at most
  $C$---however, the profit of the scheduled jobs is only guaranteed to
  be at least $\Pi - \max_{j} \pi_j$. But it is easy to fix this
  shortcoming: we choose a job $j'$ assigned by $\mathbf{S}$ to the
  virtual machine which has the largest profit, and schedule $j'$ on the
  machine where it has the least processing time. Now the modified
  assignment has cost at most $C + \max_{ij'} c_{ij'}$, makespan at most
  $T + 2 \min\{\max_{j} p_{ij}, T\}$, and the total profit of the
  scheduled jobs is at least $\Pi$.  (We assume that any job $j$ where $\min_i p_{ij} > T$ has already been discarded.)
  This is almost what we want, apart from the cost guarantee. So suppose we ``guess'' the
  ${1}/{\epsilon}$ most expensive assignments in OPT (in time
  $O(mn^{1/\epsilon})$), and hence we can focus only on the jobs having
  $c_{ij} \leq \epsilon C$ for all possible remaining assignments.  Now
  the cost of the assignment is $C + \max_{ij'} c_{ij'} \leq C(1 +
  \epsilon)$, and the makespan is at most $3T$. This completes the proof.
\end{proof}
In fact, the $(1 + \epsilon)$ loss in cost is inevitable since we can reduce the knapsack problem to the single machine makespan minimization with outliers problem: values of items become profits of jobs, and their weights become the assignment cost; the weight budget is the cost budget, and the required value is the required profit.
As for the makespan guarantee, the $3/2$-hardness of Lenstra et al.~\cite{LST90} carries over.

\section{Weighted Sum of Completion Times}
\label{sec:wsoct}

We now turn our attention to average completion time---in particular, to
$R|r_j, \out|\sum_j w_j C_j$. The main result of this section is a
constant factor approximation for this problem. Not surprisingly, the
integrality gap of standard LP relaxations is large\footnote{Implicit in
  the work of Guha and Munagala~\cite{GuhaM07} is an algorithm which
  violates the profit requirement by a constant factor; they also
  comment on the integrality gap, and pose the problem of avoiding this
  violation.}, and hence we strengthen the time-indexed formulation with
the so-called knapsack-cover inequalities~\cite{CFLP00,Wolsey75}. We
show that a randomized rounding scheme similar to that of Schulz and
Skutella~\cite{SS99} gives us the claimed guarantees on the objective
function, and while preserving the profit requirements with constant
probability.


\subsection{A Constant  Approximation for Weighted Sum of Completion Times}
We have a collection of $m$ machines and $n$ jobs, where each
job $j$ is associated with a profit $\pi_j$, a weight $w_j$, and a
release date $r_j$. When job $j$ is scheduled on machine $i$, it incurs
a processing time of $p_{ij}$. Given a parameter $\Pi > 0$, the
objective is to identify a set of jobs $S$ and a feasible schedule such
that $\sum_{j \in S} \pi_j \geq \Pi$ and such that $\sum_{j \in S} w_j
C_j$ is minimized. Here, $C_j$ denotes the completion time of job $j$.

\subsubsection{A Time Indexed LP Relaxation}

For the non-outlier version, in which all jobs have to be scheduled, Schulz and Skutella~\cite{SS99} gave a constant factor approximation by making use of a time-indexed LP. We first describe a natural extension of their linear program to the outlier case, while also \emph{strengthening} it.
\[ \begin{array}{lll}
\mbox{minimize} & {\sum_{j = 1}^n w_j C_j} \\
\mbox{subject to} & (1) \quad { C_j = \sum_{i=1}^{m} \sum_{t=0}^{T} \left( \frac{x_{ijt}}{p_{ij}} \left(t + \frac{1}{2}\right) + \frac{x_{ijt}}{2} \right)} \qquad & \forall \, j \\
& (2) \quad { y_j = \sum_{i=1}^{m} \sum_{t=0}^{T} \frac{x_{ijt}}{p_{ij}}} & \forall \, j \\
& (3) \quad { \sum_{j = 1}^n x_{ijt} \leq 1} & \forall \, i, t \\
& (4) \quad { \sum_{j \notin {\cal A}} \pi^{\cal A}_{j} y_{j}} \geq \Pi - \Pi( {\cal A} ) & \forall \, {\cal A} : \Pi({\cal A}) < \Pi \\
& (5) \quad x_{ijt} = 0 & \forall \, i, j, t : t < r_j \\
& (6) \quad x_{ijt} \geq 0, \, 0 \leq y_j \leq 1 & \forall \, i, j, t
\end{array} \]%
In this formulation, the variable $x_{ijt}$ stands for the fractional amount of
time machine $i$ spends on processing job $j$ in the time interval
$[t,t+1)$; note that the LP schedule may be preemptive. The variable
$C_{j}$, defined by constraint (1), is a measure for the completion time
of job $j$. In any integral solution, where job $j$ is scheduled from
$t$ to $t+p_{ij}$ on a single machine $i$, it is not difficult to verify
that $C_{j}$ evaluates to $t+p_{ij}$. The variable $y_{j}$, defined by
constraint (2), is the fraction of job $j$ being scheduled. Constraint
(3) ensures that machine $i$ spends at most one unit of processing time
in $[t,t+1)$. Constraints (5) and (6) are additional feasibility
checks.

We first observe that replacing the set of constraints (4) by a single
inequality of the form $\sum_{j = 1}^n \pi_j y_j \geq \Pi$ would result in an
unbounded integrality gap -- consider a single job of profit $M$, and $\Pi =
1$; the LP can schedule a $1/M$ fraction of the job, incurring a cost
which is only $1/M$ times the optimum. We therefore add in the family of constraints (4), known as the \emph{knapsack-cover (KC) inequalities}. Let ${\cal A}$ be any set
of jobs, and let $\Pi({\cal A}) = \sum_{j \in {\cal A}} \pi_j$ be the
sum of profits over all jobs in ${\cal A}$. Then, $[\Pi - \Pi({\cal
  A})]^+$ is the profit that needs to be collected by jobs not in ${\cal
  A}$ when all jobs in ${\cal A}$ are scheduled. Further, if ${\cal A}$
does not fully satisfy the profit requirement, any job $j \notin {\cal
  A}$ has a marginal contribution of at most $\pi^{\cal A}_{j} = \min \{
\pi_{j}, \Pi - \Pi({\cal A}) \}$. Therefore, for every set ${\cal A}$
such that $\Pi({\cal A}) < \Pi$, we add a constraint of the form
$\sum_{j \notin {\cal A}} \pi^{\cal A}_{j} y_{j} \geq \Pi - \Pi( {\cal
  A} )$. Note that there are exponentially many such constraints, and
hence we cannot naively solve this LP.

\medskip \noindent {\bf ``Solving'' the LP.} We will not look to find an
optimal solution to the above LP; for our purposes, it suffices to
compute a solution
vector $(\widehat{x}, \widehat{y}, \widehat{C})$ satisfying the following:
\begin{OneLiners}
\item[(a)] Constraints (1)-(3) and (5)-(6) are satisfied.

\item[(b)] Constraint (4) is satisfied for the single set $\{ j : \widehat{y}_j \geq 1/2 \}$.

\item[(c)] $\sum_{j = 1}^n w_j \widehat{C}_j \leq 2 \cdot \Opt$, where $\Opt$ denotes the cost of an optimal integral solution.
\end{OneLiners}
\vspace{5pt}
We compute this solution vector by first guessing $\Opt$ up to a
multiplicative factor of $2$ (call the guess $\widetilde{ \Opt }$), and
add to the LP the explicit constraint $\sum_{j = 1}^n w_j C_j \leq
\widetilde{ \Opt }$. Then, we solve the LP using the ellipsoid
algorithm. For the separation oracle, in each iteration,
we check if the current solution satisfies
properties (a)-(c) above. If none of these properties is violated, we
are done; otherwise, we have a violated constraint. We now present our rounding algorithm (in Algorithm 1), based on the non-outlier algorithm of \cite{SS99}.

\begin{algorithm}[h!]
  \label{alg:wjcj}
  \caption{Weighted Sum of Completion Times}
  \begin{algorithmic}[1]
\STATE \textbf{given} a solution vector $(\widehat{x}, \widehat{y}, \widehat{C})$ satisfying
properties (a)-(c), \textbf{let} ${\cal A}^*$ be the set $\{j \, | \, \widehat{y}_{j} \geq 1/2 \}$.
\STATE \textbf{for} each job $j$, do the following steps
\begin{OneLiners}

\item[{\footnotesize 2a:}] \textbf{if} $j \in {\cal A}^*$, for each $(i,t)$ pair,
  set $l_{ijt} = \widehat{x}_{ ijt } / ( p_{ ij } \widehat{y}_j )$. Note that for such jobs  $j \in {\cal A}^*$, we have $\sum_{i = 1}^m \sum_{t = 0}^T l_{ijt} = 1$ from constraint (2) of the LP.

\item[{\footnotesize 2b:}] \textbf{if} $j \notin {\cal A}^*$,
  set $l_{ijt} = 2 \widehat{x}_{ ijt } / p_{ ij }$. In this case,
  note that $\sum_{i = 1}^m \sum_{t = 0}^T l_{ijt} = 2
  \widehat{y}_j$. 

\item[{\footnotesize 2c:}] \label{algwjcj:step2c} \textbf{partition} the interval $[0,1]$ in the following way: assign each $(i,t)$ pair a sub-interval $I_{it}$ of $[0,1]$  of length $l_{ijt}$ such that these sub-intervals are pairwise disjoint. Then choose a uniformly random number $r \in [0,1]$ and set $\tau_j$ to be the $(i,t)$ pair s.t $r \in I_{it}$. If there is no such $(i,t)$ pair, leave $j$ unmarked.
\end{OneLiners}
\STATE \textbf{for} each machine $i$, consider the jobs such that $\tau_j = (i, *)$;
order them in increasing order of their marked times; schedule them
as early as possible (subject to the release dates) in this order.
\end{algorithmic}
\end{algorithm}

\subsection{Analysis}
We now show that the expected weighted sum of completion times is $O(1) \Opt$, and also that with constant probability, the total profit of the jobs scheduled is at least $\Pi$.
\begin{lemma}
\label{lem:exp-cost}
The expected weighted sum of completion times is at most $16 \cdot \Opt$.
\end{lemma}
\begin{proof}
Let $C^{R}_{j}$ be a random variable, standing for the completion time of job $j$; if this job has not been scheduled, we set $C^R_j = 0$. Since $\sum_{j = 1}^n w_j \widehat{C}_j \leq 2 \cdot \Opt$, it is sufficient to prove that $\expar{C^R_{j}} \leq 8 \widehat{C}_j$ for every $j$. To this end, note that
{\small
\begin{MyEqn}
\ex{C^R_{j}} &=& \sum_{i=1}^{m} \sum_{t = 0}^{T} \Pr [\tau_j = (i,t) ] \cdot \ex{C^R_{j} | \tau_j = (i,t)}
  \leq \sum_{i=1}^{m} \sum_{t = 0}^{T} \frac{2\widehat{x}_{ijt}}{p_{ij}}
  \cdot \ex{C^R_{j} | \tau_j = (i,t)} \ ,
\end{MyEqn}%
}
where the last inequality holds since $\Pr [\tau_j = (i,t) ] = l_{ijt} \leq 2 \widehat{x}_{ijt} / p_{ij}$, regardless of whether $j \in {\cal A}^*$ or not. Now let us upper bound $\expar{C^R_{j} | \tau_j = (i,t)}$.
The total time for which job $j$ must wait before being processed on machine $i$ can be split in the worst case into: (a) the idle time on this machine before $j$ is processed, and (b) the total processing time of other jobs marked $(i,t')$ with $t' \leq t$. If job $j$ has been marked $(i,t)$, the idle time on machine $i$ before $j$ is processed is at most $t$. In addition, the total expected processing time mentioned in item (b) is at most
{\small
\begin{MyEqn}
& & \sum_{k \neq j} p_{ik} \sum_{t' = 0}^t \pr{ \left. \tau_k = (i,t') \right| \tau_j = (i,t) } = \sum_{k \neq j} p_{ik} \sum_{t' = 0}^t \pr{ \tau_k = (i,t') }\\ & & \qquad \qquad  \leq \sum_{k \neq j} p_{ik} \sum_{t' = 0}^t \frac{2 \widehat{x}_{ikt'}}{p_{ik}}
 = 2 \sum_{t' = 0}^t \sum_{k \neq j} \widehat{x}_{ikt'} \leq 2 (t+1) \ ,
\end{MyEqn}%
}
where the last inequality follows from constraint~(3). Combining these observations and constraint (1), we have
{\small
\begin{MyEqn}
\ex{C^R_{j}} & \leq & 2 \sum_{i=1}^{m} \sum_{t = 0}^{T} \frac{\widehat{x}_{ijt}}{p_{ij}} \left( t + 2(t+1) + p_{ij} \right)
 \leq  8 \sum_{i=1}^{m} \sum_{t=0}^{T} \left( \frac{\widehat{x}_{ijt}}{p_{ij}} \left(t + \frac{1}{2}\right) + \frac{\widehat{x}_{ijt}}{2} \right)
 =  8 \widehat{C}_j \ .
\end{MyEqn}%
}
\vspace{-0.1in}
\end{proof}

\begin{lemma}
  \label{lem:probability-wjcj}
  The randomized rounding
  algorithm produces a schedule that meets the profit constraint with probability at least $1/5$.
\end{lemma}
\begin{proof}
  Clearly, when the jobs in ${\cal A}^*$ collectively satisfy the profit
  requirement, we are done since the algorithm picks every job
  in this set. In the opposite case, consider the Knapsack Cover
  inequality for ${\cal A}^*$, stating that $\sum_{j \notin {\cal A}^*}
  \pi^{{\cal A}^*}_{j} \widehat{y}_{j} \geq \Pi - \Pi( {\cal A}^* )$.
The total profit collected from
  these jobs can be lower bounded by $Z = \sum_{j \notin {\cal A}^*}
  \pi^{{\cal A}^*}_{j} Z_{j}$, where $Z_j$ is a random variable
  indicating whether job $j$ is picked.

Since our rounding algorithm picks all jobs in ${\cal A}^*$, the profit requirement is met if $Z$ is at least $\Pi - \Pi( {\cal A}^*)$. To provide an upper bound on the probability that $Z$ falls below $\Pi - \Pi( {\cal A}^*)$,
notice that by the way the algorithm marks jobs in Step~\ref{algwjcj:step2c}, we have that each job not in ${\cal A}^*$ is marked with probability $2 \widehat{y}_{j}$, \emph{independently} of the other jobs.
Therefore,
{\small
\[ \ts \ex{Z} = \mathrm{E} \big[ \sum_{j \notin {\cal A}^*} \pi^{{\cal
    A}^*}_{j} Z_{j} \big] = 2 \sum_{j \notin {\cal A}^*} \pi^{{\cal A}^*}_{j} \widehat{y}_{j} \geq 2(\Pi - \Pi( {\cal A}^*)) \ . \]%
}
Consequently, if we define $\alpha_j = \pi^{{\cal A}^*}_{j} / (\Pi - \Pi( {\cal A}^*))$, then
{\small
\begin{eqnarray}
\nonumber \pr{ Z \leq \Pi - \Pi( {\cal A}^*) } & = & \ts \pr{ \sum_{j \notin {\cal A}^*} \frac{ \pi^{{\cal A}^*}_{j} }{ \Pi - \Pi( {\cal A}^*) } Z_{j} \leq 1 } \leq \pr{ \sum_{j \notin {\cal A}^*} \alpha_j Z_{j} \leq \frac{ \ex{ \sum_{j \notin {\cal A}^*} \alpha_j Z_{j} } }{ 2 } } \\
\nonumber & \leq & \ts \exp \big( - \frac{ 1 }{ 8 } \cdot \mathrm{E} \big[ \sum_{j \notin {\cal A}^*} \alpha_j Z_{j} \big] \big) \leq e^{ -1/4 } < \frac{ 4 }{ 5 } \ ,
\end{eqnarray}
} where the first and third inequalities hold since $\expar{ \sum_{j
    \notin {\cal A}^*} \alpha_j Z_{j} } \geq 2$, and the second
inequality follows from bounding the lower tail of the sum of
independent $[0,1]$ r.v.s (see, e.g., \cite[Thm.\
3.5]{chernoff}). 
\end{proof}

The above two lemmas combine to give the following theorem. 

\begin{theorem}
  \label{thm:main-wjcj}
  For  $R|r_j, \out|\sum_j w_j C_j$, there is a randomized
  $O(1)$-approximation algorithm.
\end{theorem}

While the LP formulation as stated has exponentially many time intervals of length $1$, we can make our algorithm fully polynomial in the size of the input (with a small loss in approximation guarantee) by considering geometrically increasing sizes \cite{HallSW96} for the time intervals.

In \lref[Appendix]{sec:proofs-wjcj-ext}, we show that given $K$ \emph{different} profit requirements, our algorithm can be modified to give an $O(\log K)$-approximation.

\subsection{Single Machine, Identical Weights} \label{sec:onemachine}
In this section, we show how we can get an FPTAS using dynamic programming for the problem of minimizing the unweighted sum of completion times on a constant number of machines. For simplicity, we first give the complete proof for the case of a single machine, and sketch how to extend it for a constant number of machines.

\medskip \noindent \textbf{Single Machine, Identical Weights.}
We are given a collection of $n$ jobs where job $j$ is associated with a processing time $p_j$
and a profit $\pi_j$. Given a target profit of $\Pi > 0$, the
goal is to identify a set of jobs $S$ and a corresponding
single-machine schedule such that $\sum_{j \in S} \pi_j \geq \Pi$ and
$\sum_{j \in S} C_j$ is minimized (where $C_j$ is the
completion time of job $j$).

\medskip \noindent {\bf Dynamic program.} Suppose $p_1, \ldots, p_n$ are integers such that $p_1 \leq \cdots \leq p_n$. Let $\profit(j,C,L)$ be the maximum  profit that can be collected by scheduling a subset of jobs $\{1 , \ldots, j\}$ such that their sum of completion times is \emph{at most} $C$ and makespan is \emph{exactly} $L$. Then, the following recurrence holds:
\vspace{-5pt}
\[ \profit(j,C,L) = \max \{ \profit(j-1,C,L) , \pi_j + \profit(j-1, C - L, L - p_j) \} \]

\vspace{-5pt}
\noindent To better understand the above equation, notice that if job $j$ is picked by an optimal schedule, it will not be scheduled before any of the jobs $\{1, \ldots, j-1\}$ since the {\em shortest processing time} strategy is optimal for a fixed set of jobs (\cite{KargerSW97}).
Therefore, consider a set of jobs $\{1, \ldots, j\}$ that have a bound $C$ on their sum of completion times and let $L$ be their makespan.
If $j$ is scheduled, the jobs $\{1, \ldots, j-1\}$ \emph{must} have a residual makespan of $L-p_j$ and a bound of $C-L$ on the sum of completion times since job $j$ incurs a completion time of $L$ by virtue of it being scheduled last among $\{1, \ldots, j\}$; we
 also collect a profit of $\pi_j$ in this case. On the other hand, if $j$ is not scheduled, $C$ and $L$ remain the same but we don't collect any profit. Now, given this recurrence, the goal is to find the minimal $C$ and some $L$ such that $\profit(n, C, L) \geq \Pi$. This can be solved by dynamic programming, with running time $O(n C_{max} L_{max})$. Since $C \leq n^2 p_n$ and $L \leq n p_n$, the running time is $O( n^4 p_n^2)$, i.e. pseudo-polynomial.


Therefore, the above dynamic program can be used to compute an optimal solution in polynomial time when all processing times are small integers. We now apply scaling techniques to obtain an FPTAS to handle arbitrary processing times.

\medskip \noindent {\bf Handling general instances.}
 Given an instance ${\cal I}$ of the original problem, we begin by ``guessing'' $P_{ \max }$, the maximum processing time of a job that is scheduled in some fixed optimal solution. We now create a new instance ${\cal I}'$ in which every job $j$ with $p_j > P_{ \max }$ is discarded; other jobs get a scaled processing time of $p_j' = \lceil p_j / K \rceil$, where $K = (2\epsilon P_{\max}) / (n(n+1))$. Notice that the scaled processing times of remaining jobs are integers in $[0, \lceil n(n+1)/(2\epsilon) \rceil]$. We can therefore find in $O( n^8 / \epsilon^2 )$ time an optimal subset of jobs ${\cal J}_{{\cal I}'}$ to be scheduled in ${\cal I}'$, and return this set as a solution for ${\cal I}$.

\begin{theorem}
Scheduling the jobs ${\cal J}_{{\cal I}'}$ in order of non-decreasing processing times guarantees that their sum of completion times is at most $(1 + \epsilon) \Opt( {\cal I} )$.
\end{theorem}
\begin{proof}
We begin by relating $\Opt( {\cal I} )$ to $\Opt( {\cal I}' )$. For this purpose, suppose that ${\cal J}_{{\cal I}} = \{ j_1, \ldots, j_R \}$ in an optimal solution to ${\cal I}$. Then,
\[ \Opt( {\cal I}' ) \leq \sum_{r = 1}^R \sum_{s = 1}^r p_{j_s}' \leq \sum_{r = 1}^R \sum_{s = 1}^r \left( \frac{ p_{j_s} }{ K } + 1 \right) \leq \frac{ \Opt( {\cal I} ) }{ K } + \frac{ n( n+1 ) }{ 2 } \ . \]%
Now suppose that ${\cal J}_{{\cal I}'} = \{ j_1', \ldots, j_Q' \}$. Then, the sum of completion times that results from scheduling $j_1', \ldots, j_Q'$ in this exact order is
\begin{MyEqn}
\sum_{r = 1}^Q \sum_{s = 1}^r p_{j_s'} & \leq & K \sum_{r = 1}^Q \sum_{s = 1}^r p_{j_s'}' \\
& = & K \cdot \Opt({\cal I}') \\
& \leq & \Opt( {\cal I} ) + \frac{ n( n+1 ) K }{ 2 } \\
& = & \Opt( {\cal I} ) + \epsilon P_{ \max } \\
& \leq & (1 + \epsilon) \Opt( {\cal I} ) \ .
\end{MyEqn}%
where the last inequality holds since $P_{ \max }$ is a lower bound on $\Opt( {\cal I} )$.
\end{proof}

\noindent \textbf{Constant number of machines.}

We finally consider the case when there is a constant number of identical machines, say $m$. To this end, let $\profit(j,C,L_{1},L_{2},\ldots,L_{m})$ be the maximum  profit that can be collected by scheduling a subset of the jobs $1 , \ldots, j$ such that their sum of completion times is at most $C$ and such that the makespan is exactly $L_{i}$ on machine $i$. Then $\profit(j,C,L_{1},\ldots,L_{m})$ can be written as
\[\max \left\{ \profit(j-1,C,L_{1},\ldots,L_{m}) , \max_{i} \left(\pi_j + \profit(j-1, C - L_{i}, L_{1},\ldots, L_{i} - p_j,\ldots, L_{m})\right) \right\} \ . \]%
When $m = O(1)$, the size of this dynamic program is still polynomial in $n$. The remaining analysis is similar to the one for a single machine.

\section{Minimizing Average Flow Time on Identical Machines}
\label{flowtime}

Finally, we consider the problem of minimizing the average (preemptive)
flow time on identical machines ($P | r_{j}, pmtn, \out | \sum F_j$)
with \emph{unit profits}. We present an LP rounding algorithm that
produces a preemptive non-migratory (no job is scheduled on multiple machines) schedule whose flow time is
within $O(\log P)$ of the optimal, where $P$ is the ratio
between the largest and smallest processing times.

This is the technical heart of the paper; in sharp contrast to the
problems in the previous two sections, it is not clear how to easily
change the existing algorithms for this problem to handle the outliers
case---while we use the same LP as in previous works, our LP rounding
algorithm for the outlier case has to substantially extend the previous
non-outlier rounding algorithm. Since our algorithms are somewhat
involved, we first present the algorithm for a \emph{single machine},
and subsequently sketch how to extend it to multiple identical machines.
For the rest of this section, consider the following setup: we are given
a single machine and a collection of $n$ jobs where each job $j$ has a
release date $r_j \in \mathbb{Z}$ and a processing time $p_j \in
\mathbb{Z}$.  Given a parameter $\Pi > 0$, we want to identify a set of
jobs $S$ and a preemptive schedule minimizing $\sum_{j \in S} F_j$
(where $F_j = C_j - r_j$) subject to $|S| \geq \Pi$.


\subsection{The Flow-time LP Relaxation and an Integrality Gap} \label{subsec:prelim}


Our LP relaxation is a natural outlier extension of one used in earlier
flow-time algorithms (\cite{GargK06,GargK07}).  We first describe what
the variables and constraints correspond to: \emph{(i}) $f_j$ is the
fractional flow time of job $j$, \emph{(ii)} $x_{jt}$ is the fraction of
job $j$ scheduled in the time interval $[t,t+1)$, and \emph{(iii)} $y_j$
is the fraction of job $j$ scheduled. Constraint (1) keeps track of the
flow time of each job, while constraints (2), (3), and (4) are to make
sure the solution is feasible with respect to the profit constraint.
Notice that in constraint (1), we use the quantity $\widetilde{p}_{j}$
(which denotes the processing time $p_{j}$ rounded up to the next power
of $2$), instead of $p_{j}$. Also, this modification is present only in
constraint (1) which dictates the LP cost, and not in constraint (2)
which measures the extent to which each job is scheduled. The quantity
$T$ is a guess for the time at which the optimal solution completes
processing jobs (in fact, any upper bound of it would suffice). We also
assume that a parameter $k^* \in \mathbb{Z}$ was guessed in advance,
such that the optimal solution only schedules jobs with $p_j \leq
2^{k^*}$. Our algorithm would have running time which is polynomial in
$T$ and $n$.

\[ \begin{array}{lll}
\mbox{minimize} & { \sum_{j = 1}^{n} f_j} \\
\mbox{subject to} & (1) \quad { \displaystyle f_{j} = \sum_{t=0}^{T} \left( \frac{x_{jt}}{\widetilde{p}_{j}} \left(t + \frac{1}{2} - r_{j}\right) + \frac{x_{jt}}{2} \right)} \qquad & \forall \, j \\
& (2) \quad { p_j y_j = \sum_{t=0}^{T} x_{jt}} & \forall \, j \\
& (3) \quad { \sum_{j = 1}^n x_{jt} \leq 1} & \forall \, t \\
& (4) \quad { \sum_{j = 1}^n y_{j} \geq \Pi} \\
& (5) \quad x_{jt} = 0 & \forall \,  j, t : t < r_j \\
& (6) \quad x_{jt} \geq 0, \, 0 \leq y_j \leq 1 & \forall \, j, t
\end{array} \]%

Given the above LP, we first claim that it is indeed a
\emph{relaxation}.
\begin{lemma}[Relaxation]
  \label{lem:flowlp_vs_opt}
  $\opt{ LP } \leq \Opt$, where $\Opt$ denotes the optimal sum of flow times.
\end{lemma}

\begin{proof}
Given an optimal solution for the given instance, we construct a corresponding LP solution in a natural way, by setting $x_{jt} = \Delta p_{j}$ when the optimal solution schedules a $\Delta$ fraction of job $j$ in the time interval $[t,t+1)$. It is easy to verify that the profit constraint is satisfied. Now, consider a particular job $j$ scheduled in the optimal solution. We proceed by showing that the term $f'_{j} = \sum_{t=0}^{T} ( \frac{x_{jt}}{p_{j}} (t + \frac{1}{2} - r_{j}) + \frac{x_{jt}}{2} )$ is a lower bound on the flow time of $j$ (notice that $f'_{j}$ has $p_j$ in the denominator where $f_j$ had $\widetilde{p_{j}}$).

Suppose the optimal solution completes processing $j$ at $C_{j}$. The flow time is therefore $C_{j} - r_{j}$, whereas the worst case for the LP is when $j$ is contiguously scheduled in the time interval $[C_{j} - p_{j},C_{j})$; otherwise, some fraction is scheduled earlier, and the contribution to $f'_{j}$ can only decrease. Consequently,
\[ f'_{j} \leq \sum_{t=C_{j} - p_{j}}^{C_{j}-1} \frac{t + 1/2 - r_{j}}{p_{j}} + \frac{p_{j}}{2} = C_{j}  - r_{j} \ . \]%
Since $f_{j} \leq f'_{j}$, we have $f_{j} \leq C_{j} - r_{j}$. The lemma follows by summing  over all jobs scheduled by OPT.
\end{proof}

Before getting into the details of our algorithm, to gain more intuition for this relaxation, we demonstrate that it has an integrality
gap of $\Omega(\log P)$, where $P$ is the ratio between the largest and
smallest processing times in an optimal solution.
\begin{theorem}[Integrality Gap]
  \label{thm:intgap}
  There are instances in which $\Opt = \Omega( \log P ) \cdot
  \Opt(\rm LP)$.
\end{theorem}

\begin{proof}
  Consider an instance where there are $k+1$ {\em large} jobs  numbered
  $1, 2, \ldots, k+1$. Jobs $1, 2,\ldots, k$ have processing times $2^2, 2^3, \ldots, 2^{k+1}$
  respectively and job $k+1$ has a processing time of $2^{k+1}$. In addition, there are $M = M(k)$ {\em small} jobs of
  unit processing time, where $M$ is a parameter whose value will be
  determined later. Large jobs $1,2, \ldots, k$ arrive in decreasing order of processing
  time, where job $j$ arrives at the beginning of the white block
  numbered $j$ in \lref[Figure]{integrality-gap}. White block $j$ occupies
  $2^j$ time units. Job $k+1$ arrives at the beginning at the white block numbered $k+1$ which occupies $2^{k+1}$ time units. There is also a large grey block occupying $M$ time
  units; the arrivals of small jobs are uniformly spaced in this block
  (starting at the left endpoint) with a gap of $1$. Now suppose we are
  required to schedule $M + k/2+1$ jobs.
\begin{figure}[!htbp]
\begin{center}
\includegraphics[scale=0.55]{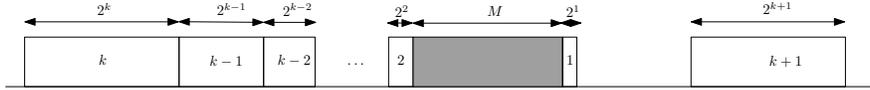}
\end{center}
\caption{{\small A schematic description of the integrality gap instance.}} \label{integrality-gap}
\end{figure}

\noindent
We first observe that the optimal schedule picks every small job (as well as $k/2+1$ large jobs). To see why, suppose one or more small jobs have not been picked, let $q$ be the minimal index of such a job.
Note that once we pick a subset of $M + k/2+1$ jobs, an optimal schedule is determined by employing the {\em shortest remaining processing time} rule (see, for example, \cite{Baker74}).
Further, from the sizes of the white blocks, we see that even if a large job is scheduled without being preempted since its release date, it would have a remaining processing time of at least $2$ at the beginning of the grey block. Therefore, from the SRPT rule, it is clear that the first $q-1$ small jobs are scheduled in the first $q-1$ time units of the grey block.
Thus, at the point when job $q$ is released, any large job has a remaining processing time of at least $2$. It follows that, by picking $q$ and dropping some large job, we can obtain a smaller flow time, implying that the schedule under consideration cannot be optimal.

Further, it is optimal to schedule large job $k+1$. If there is a solution which does not, we can schedule it while skipping one additional large job out of jobs $1,2,\ldots,k$ to improve on the average flow time (this holds if $M \geq 2^{k+1}$, which we will ensure later). Therefore, the value of $k^*$ in the LP (which denotes the largest class scheduled in an optimal solution) will be $k+1$.

Based on the above SRPT observation, we can conclude that each small job will be contiguously processed to completion immediately after its release date, and as a result no large job (from the set of jobs $1, 2, \ldots, k$) can be completed any sooner than the right endpoint of the grey block. Hence, every large job picked incurs a flow time of at least $M$, meaning that $\Opt \geq Mk/2$. On the other hand, a fractional solution can fully schedule \emph{every} small job as soon as it arrives, and schedule \emph{half} of each large job $j$ into white block $j$. It can also schedule large job $k+1$ completely into white block $k+1$. It is not difficult to verify that the cost of this solution is at most $M + \sum_{j = 1}^k 2^j + 2^{k+1} < M + 2^{k+2}$. Therefore, by setting $M = 2^{k+1}$, we have
{\small
\[ \frac{ \Opt }{ \Opt(\rm LP) } \geq \frac{ Mk }{ 2(M + 2^{ k+2 }) } = \frac{ 2^{k+1} k }{ 2(2^{k+1} + 2^{ k+2 }) } = \frac{ k }{ 6 } = \Omega( \log P ) \ , \]%
}
where the last equality holds since $P = 2^{k+1}$.
\end{proof}

Note that this gap instance is on a single machine, for which we know
that the shortest remaining processing time policy (SRPT) is optimal in the non-outlier case. However our results
eventually show that this is as bad as it gets---we show an upper bound
of $O(\log P)$ for the integrality gap even for identical machines!

\subsection{The Flow-time Rounding Algorithm: General Game Plan and Some Hurdles}
\label{sec:plan}

Before we present our algorithm in detail, let us give
a high-level picture and indicate some of the complicating factors over
the earlier work. Previous LP-based rounding techniques
\cite{GargK06,GargK07} relied on the fact that if we rearrange the jobs
of length roughly $2^k$---call such jobs ``class-$k$'' jobs---among the
time slots they occupy in the fractional solution, the objective
function does not change much; these algorithms then use this
rearrangement to make the schedule feasible (no job simultaneously
scheduled on two machines) and even non-migratory across machines. We are currently
considering the single machine case, so these issues are irrelevant for
the time being (and we will come back to them later)---however, we need
to handle jobs that are fractionally picked by the LP. In particular, we
need to swap ``mass'' between jobs to pick an integral number of jobs to
schedule. And it is this step which increases the LP cost even in the
case of a single machine. Note that we essentially care only about the
$y_j$ value for each job $j$, which indicates the extent to which this
job is scheduled---if we could make them integral without altering the
objective by much, we would be done!

However, na{\"i}ve approaches to make the $y_j$'s integral may have bad
approximation guarantees. E.g., consider taking two consecutive
fractional jobs $j$ and $j'$ with similar processing times (observe that
jobs with similar processing times have similar contributions to the
objective, except for the release date component) and scheduling more of
the first one over the second. If the second job $j'$ has \emph{even
  slightly smaller} processing time than $j$ has, we would run out of
space trying to schedule an equal fraction of $j$ over $j'$, and this
loss may hurt us in the (hard) profit requirement. In such a case, we
could try to schedule $j'$ over $j$, observing that the later job $j'$
would not advance too much in time, since $j$ and $j'$ were consecutive
in that class and have similar processing times---the eventual hope
being that given a small violation of the release dates, we may be able
to shift the entire schedule by a bit and regain feasibility.

But this strategy could lead to arbitrarily bad approximations: we could
keep fractionally growing a job $j$ until (say) $2/3$ of it is
scheduled, only to meet a job $j'$ subsequently that also has $2/3$ of
it scheduled, but $j'$ has smaller processing time and therefore needs
to be scheduled over $j$. In this case, $j$ would shrink to $1/3$, and
then would start growing again---and repeated occurrences of this might
cause the flow time for $j$ to be very high. Indeed, trying to avoid
such situations leads us to our algorithm, where we look at a
\emph{window} of jobs and select an appropriate one to schedule, rather
than greedily running a swapping process. To analyze our algorithm, we
charge the total increase in the fractional flow time to the
\emph{fractional makespan} of the LP solution, and show that each class
of jobs charges the fractional makespan  at most twice.

\subsection{Notation and Preliminaries}

We partition the collection of jobs into \emph{classes}, with jobs in
class ${\cal C}_k$ having $p_{j} \in (2^{k-1}, 2^k]$. Notice that
$\widetilde{p}_{j} = 2^k$ for every $j \in {\cal C}_k$, and the class of
interest with highest index is $\mathcal{C}_{k^*}$. Given a fractional
solution $(x,y,f)$, we say that job $j$ is {\em fully scheduled} if $y_j
= 1$, and {\em dropped} if $y_j = 0$; in both cases, $j$ is {\em
  integrally scheduled}. Let $\flow( x, y, f ) = \sum_{j = 1}^n f_j$ be
the fractional cost; note that this is \emph{not the same} as the actual
flow time given by this solution, but rather an approximation. Let
${\cal P}(x,y,f) = \sum_{j = 1}^n \sum_{t = 0}^T x_{jt}$ be the total
fractional processing time.  Since each job $j$ gets $x_{jt}$ amount of
processing time in $[t,t+1)$, the cost of $(x,y,f)$ remains unchanged if
all jobs are processed during the first part $[t, t + \sum_{j = 1}^n
x_{jt})$ of this unit interval; we therefore refer to $[t + \sum_{j =
  1}^n x_{jt}, t+1)$ as the {\em free time interval} in $[t,t+1)$.


We say that an LP solution $(x,y,f)$ is ``non-alternating'' across each class if
the fractional schedule does not alternate between two jobs of the
same class. Formally, the schedule is ``non-alternating'' if for class
$k$ and any two class-$k$ jobs $j$ and $j'$, if $y_j, y_{j'} >
0$ and $r_j < r_{j'}$ (or $r_j = r_{j'}$ and $j < j'$), then for any
times $t, t'$ such that $x_{jt} > 0$ and $x_{j't'} > 0$, it holds that
$t \leq t'$. We call a solution ``packed'' if there is no free time
between the release date of a job, and the last time it is scheduled by
the LP solution. The following lemma is proved in
\lref[Appendix]{sec:flowalgo-proofs-I}; we assume that we start off with
such a solution.
\begin{lemma}
  \label{lem:non-alt}
  There is an optimal LP solution $(x^*,y^*,f^*)$ that is
  non-alternating and packed.
\end{lemma}

\subsection{The Flow-Time Rounding Algorithm}

At a high level, the rounding algorithm proceeds in two stages.
\begin{OneLiners}
\item In Stage~I, for each $k$, we completely schedule
  almost as many class-$k$ jobs as the LP does fractionally (up to an additive two jobs). The main challenge, as sketched above, is to do
  this with only a small change in the fractional flow time and the
  processing time of these jobs.

\item In Stage~II, we add in at most two class-$k$ jobs to
  compensate for the loss of jobs in Stage~I. Since we add only two
  jobs per class, we can show that the additional flow time
  can be controlled.
\end{OneLiners}

\subsubsection{Flow-Time Rounding: Stage I}
\label{sec:stage1}

Recall that we want to convert the non-alternating and packed optimal
solution $(x^*, y^*, f^*)$ returned by the LP into a new solution $(x',
y', f')$ where at least $\lfloor \smash{\sum_{j \in {\cal C}_k}} y^*_j
\rfloor - 1$ class-$k$ jobs are \emph{completely scheduled}. The
algorithm operates on the classes one by one. For each class, it
performs a \emph{swapping phase} where mass is shifted between jobs in
this class (potentially violating release dates), and then does a
\emph{shifting phase} to handle all
the release-date violations. 

\medskip \noindent {\bf Swapping Phase for Class-$\bs{k}$.}
Given the non-alternating and packed solution $(x^*, y^*, f^*)$, we run the algorithm for the swapping phase given in Algorithm $2$.

\begin{algorithm}[h!]
  \caption{Class-$k$ Swapping}
  \label{alg:stage2}
  \begin{algorithmic}[1]
\STATE  \label{swap:step1} \textbf{set} $(x', y', f') := (x^*, y^*, f^*)$. Repeat the steps \ref{swap:step2}-\ref{swap:step5} until
$\lfloor \smash{\sum_{j \in {\cal C}_k}} y^*_j \rfloor - 1$ class-$k$ jobs are completely scheduled in $(x', y', f')$.
\STATE  \label{swap:step2}\textbf{advance} all class-$k$ jobs as much as possible without violating
  release dates (for jobs already violating release dates, don't advance their starting time any further) within the time intervals that are either free or are occupied by class-$k$ jobs.
\STATE  \label{swap:step3}\textbf{let} $j_1$ the first fractionally scheduled job in the current LP
  solution $(x', y', f')$. Let $j_{q+1}$ be the first class-$k$ job scheduled after
  $j_1$ which has processing time $p_{j_{q+1}} < p_{j_1}$, and say the
  class-$k$ jobs that are scheduled between $j_1$ and $j_{q+1}$ are
  $j_2, j_3, \ldots, j_q$. Note that all these jobs must have greater
  processing time than $p_{j_1}$. Also, let $\free$ denote the total
  free time between $j_1$ and $j_{q+1}$ in the current schedule.
\STATE  \label{swap:step4} \textbf{if} $\sum_{k=2}^{q} y'_{k} + \free/{p_{j_1}} \geq 1 - y'_{j_1}$,
  then we know that $j_1$ can be \emph{completely} scheduled over the
  jobs $j_2, j_3, \ldots, j_q$ and the free time; \textbf{for} $k = 2$ to $q$, do the following
\begin{OneLiners}
\item[{\footnotesize 4a:}] \textbf{if} there is some free time (of total length, say, $L$) between $j_{k-1}$ and $j_{k}$, schedule a fraction $\Delta = \min(1 - y'_{j_1}, L/p_{j_1})$ of $j_1$ in the free time, and delete a fraction $\Delta$ from class-$k$ jobs at the rear end of the schedule. Update $(x', y', f')$.
\item[{\footnotesize 4b:}] \textbf{schedule} a fraction $\Delta = \min(1 - y'_{j_1}, y'_{j_k})$ of $j_1$ over a fraction $\Delta$ of job $j_k$ (possibly creating some free space).  Update $(x', y', f')$.
\item[{\footnotesize 4c:}] \textbf{if} $k = q$ and there is some free time (of total length, say, $L$) between $j_{q}$ and $j_{q+1}$, schedule a fraction $\Delta = \min(1 - y'_{j_1}, L/p_{j_1})$ of $j_1$ in the free time, and delete a fraction $\Delta$ from class-$k$ jobs at the rear end of the schedule.  Update $(x', y', f')$.
\end{OneLiners}
\STATE  \label{swap:step5} \textbf{else if} $\sum_{k=2}^{q} y'_{k} + \free/{p_{j_1}} <
  1 - y'_{j_1}$, do the following
\begin{OneLiners}
\item delete a total fraction $\min(\sum_{k=1}^{q} y'_{k}, y'_{j_{q+1}})$ from a prefix of jobs $j_1, j_2, \ldots, j_q$, and \emph{advance} the current fractional schedule of the job $j_{q+1}$ to occupy the space created. Update the solution $(x', y', f')$. Note that it may or may not have been possible to
  schedule $j_1$ in the space fractionally occupied by jobs $j_2,
  j_3, \ldots, j_q$ and free time in this interval; for accounting
  reasons we do the same thing in both cases.
\end{OneLiners}
 \end{algorithmic}
\end{algorithm}

\medskip \noindent {\bf Shifting Phase for Class-$\bs{k}$.}
After the above swapping phase for class-$k$ jobs, we perform a {\em
  shifting phase} to handle any violated release dates.  Specifically,
consider the collection of time intervals occupied either by class-$k$
jobs or by free time---by the process given above, this remains fixed
over the execution of the swapping phase. We now shift all class-$k$
jobs to the right by $2 \cdot 2^k$ within these intervals. Of course, we need to prove that this takes care of all release date
violations.

\subsubsection{Analysis for Stage~I}

\begin{lemma}
\label{lem:stage1cost}
The following properties hold true at the end of Stage~I:
\begin{OneLiners}
\item [(i)] ${\cal P}(x', y', f') \leq  2{\cal P}(x^*, y^*, f^*)$
\item [(ii)] The fractional flow time satisfies $\flow(x', y', f') \leq 4 \cdot \flow(x^*, y^*, f^*) + 6 k^* {\cal P}(x^*, y^*, f^*)$.
\item [(iii)] The sum of flow times over all fully scheduled jobs is at most $2 \cdot \flow(x',y',F') + k^* {\cal P}(x',y',F')$.
\end{OneLiners}
\end{lemma}

The analysis proceeds by a somewhat delicate charging argument and the
basic idea is the following. In \lref[Step]{swap:step4} of the
algorithm, suppose $\Delta$ fraction of a job $j_1$ is being scheduled
over $\Delta$ fraction of a job $j_k$: we will \emph{charge} every point
in the interval $(r_{j_1}, r_{j_k})$ by an amount $\Delta$. In the case
when a $\Delta$ fraction of $j_1$ is being scheduled over an interval of
free time beginning at $t$, we will then charge every point in the
interval $(r_{j_1}, t)$ by the fraction $\Delta$. We then go on to show
that $\flow(x', y', f') - \flow(x^*, y^*, f^*)$ is not too much more
than the total charge accumulated by the interval $[0,T]$ (recall that
$T$ is the last time at which the LP scheduled some fractional job). To
complete the proof, we argue that the total charge accumulated is
$O(\log P) {\cal P}(x^*, y^*, f^*)$.

In \lref[Appendix]{app:stage1proof} we restate Stage~I in
a slightly different way, where we also define the charging process
$\charge$ associated with each step of the algorithm, and give the
complete proof of \lref[Lemma]{lem:stage1cost}.

\subsubsection{Flow-Time Rounding: Stage II}
\label{sec:stage2}

The fractional solution $(x',y',F')$ may not be feasible, since we have
only scheduled $\lfloor \sum_{j \in {\cal C}_k} y_j^* \rfloor - 1$ jobs
from class-$k$. Hence, for each class-$k$, arbitrarily pick the minimum
number of non-fully-scheduled jobs to bring this number to $\lceil
\sum_{j \in {\cal C}_k} y_j^* \rceil$ (at most two per
class). These jobs are preemptively scheduled as soon as possible after
their release date. Since at most two jobs per class are added,
the flow time does not change much. 

\begin{lemma} \label{lem:stage2cost}
  The sum of the flow times of all added jobs is at most $k^*(
  {\cal P}(x',y',F') + 2^{k^* + 2})$.
\end{lemma}

\begin{proof}
For a class-$k$, we may have to complete two additional jobs. When we schedule an extra job as soon as possible, it waits only for jobs that were fully scheduled during stage II or for jobs that were added in previous iterations of the current stage. Therefore, its flow time can be at most ${\cal P}(x',y',F') + 2 \sum_{k = 1}^{ k^* } 2^k$, and therefore the total flow time of added jobs is at most $k^*( {\cal P}(x',y',F') + 2^{k^* + 2})$.
\end{proof}

Because $f_j$ is lowerbounded by $\sum_{t} x^*_{jt}/2$, we have that ${\cal P}(x^*, y^*, f^*) \leq 2\cdot\Opt$.
Therefore, \noindent Lemmas \ref{lem:stage1cost} and \ref{lem:stage2cost}
in conjunction with the inequalities ${\cal P}(x',y',F') \leq 2 {\cal
  P}(x^*, y^*, f^*) \leq 4\cdot\Opt$  and $k^*
\leq \log P + 1$, prove the following result for minimizing flow time on a single machine.
\begin{theorem}
  \label{thm:main-flowtime}

  The problem of minimizing flow time on a single machine with unit
  profits can be approximated within a factor of $O(\log P)$.
\end{theorem}

\subsection{Flow-Time: Identical Parallel Machines}
\label{sec:parallel-flow}



We conclude this section by showing how to combine our single machine
algorithm along with ideas drawn from~\cite{GargK06} to obtain an
$O(\log P)$ approximation for the case of identical machines. We begin
by solving a natural extension of the single machine LP to the setting
of identical machines; let $(x^*,y^*,f^*)$ be the resulting fractional
solution.
\begin{OneLiners}
\item \textbf{Stage I:}  We rearrange the jobs to make the schedule non-migratory, while preserving the fraction to which each job has been scheduled. This modification is done by performing the procedure given in~\cite{GargK06} with the only change being that the jobs in $(x^*,y^*,f^*)$ are fractionally scheduled. The resulting solution $(\hat{x},\hat{y},\hat{f})$ can be shown to have an LP cost of at most $\flow(x^*,y^*,f^*) + O(\log P){\cal P}(x^*,y^*,f^*)$.

\item \textbf{Stage II:}  For each machine, we execute Stage I of the
  single machine algorithm. As a result, almost all jobs are now
  integrally scheduled, leaving at most two fractionally scheduled jobs
  per class and machine, without increasing the LP cost by much. If
  $(x',y',f')$ denotes the LP solution after this stage, we can show
  that $\flow(x',y',f') \leq 2\cdot\flow(x^*,y^*,f^*) + 6\log P \cdot
  {\cal P}(\hat{x},\hat{y},\hat{f})$  (The analysis is identical to that
  for Stage~I of the single machine algorithm, presented in \lref[Appendix]{sec:stage1}).

\item \textbf{Stage III:} We consider all fractionally scheduled class-$k$ jobs (there are at most $2$ per machine) and schedule more of the job with least processing time, while deleting an equal fraction from the one with largest processing time, until either the small job is fully scheduled or the large job has been completely dropped. This procedure is repeated till only at most one fractional class-$k$ job remains. The entire process is repeated for each class. In \lref[Appendix]{sec:parallel-machines-algo}, we show that the LP cost satisfies $\flow(\widetilde{x},\widetilde{y},\widetilde{f}) \leq \flow(x',y',f') + 2 \log P \cdot {\cal P}(x',y',f')$ where $(\widetilde{x},\widetilde{y},\widetilde{f})$ denotes the solution after Stage III. Subsequently, bounding the actual sum of flow times (of the integrally scheduled jobs) also follows closely to \lref[Lemma]{lem:flow_time_cost}.

\item \textbf{Stage IV:} As in the single machine case, we schedule the one remaining fractional job for each class by adding processing time whenever possible. This change does not significantly increase the solution cost much like \lref[Lemma]{lem:stage2cost}.
\end{OneLiners}
The algorithm and the proofs depend heavily on the single machine case; we give more details in
\lref[Appendix]{sec:parallel-machines-algo}. To conclude, we have the following theorem.
\begin{theorem} \label{thm:parallel-flowtime}
The problem of minimizing flow time on identical machines with unit profits can be approximated within a factor of $O(\log P)$.
\end{theorem}


\medskip \noindent {\bf Acknowledgements:} We would like to thank Nikhil Bansal, Chandra Chekuri, Kirk Pruhs, Mohit Singh, and Gerhard Woeginger for useful discussions. We would also like to thank the anonymous reviewers of an earlier version of this paper for several helpful comments.

{\small
\bibliographystyle{abbrv}
\bibliography{scheduling}
}
\appendix

\section{Proofs from Section~\ref{flowtime}}
\label{sec:proofs-flowtime}

\subsection{Proof of Lemma \ref{lem:non-alt}: Non-Alternating and Compact Optimal Solutions} \label{sec:flowalgo-proofs-I}


Consider doing the following changes for every class $1 \leq k \leq k^*$ in some arbitrary order.
{
\begin{OneLiners}
\item[1.] Let ${\cal T}(k)$ be the union of all time intervals in which
  the LP solution $(x^*, y^*, f^*)$ schedules class-$k$ jobs along with the overall free time.
\item[2.] Use ${\cal T}(k)$ to continuously schedule a $y_j^*$
  fraction of each job $j \in {\cal C}_k$. These fractions are scheduled
  in increasing order of release dates as soon as possible (while respecting release dates).
\end{OneLiners}
}
Let $(\widehat{x}, \widehat{y}, \widehat{f})$ be the resulting LP solution after we
finish this operation. Notice that this solution is non-alternating and packed. Also, every job is still
scheduled to the same extent as before, i.e. $\widehat{y}_j = y^*_{j}$, and consequently, we have, ${\cal P}(\widehat{x}, \widehat{y}, \widehat{f}) = {\cal P}(x^*,y^*,F^*)$. It remains to show that $\flow(\widehat{x}, \widehat{y}, \widehat{f}) \leq \flow( x^*,y^*, f^* )$.
Let $V_{k,t}(x, y, f) = \sum_{j \in {\cal C}_k} \sum_{t' = t}^T x_{jt' }$ be the overall processing time of class-$k$ jobs after time $t$.
Clearly since we are only rearranging class-$k$ jobs (and subsequently advancing jobs whenever possible) within the time intervals in which they were scheduled in $(x^*, y^*, f^*)$,
we have $V_{k,t}(\widehat{x}, \widehat{y}, \widehat{f}) \leq V_{k,t}( x^*,y^*,f^* )$ for every $k$ and $t$. Now,

\vspace{-10pt}
  \begin{MyEqn}
    & & \flow(\widehat{x}, \widehat{y}, \widehat{f}) - \flow( x^*,y^*,f^* ) = \sum_{j = 1}^n \left( \widehat{f}_j - f_j^* \right) \\
    & &  \qquad = \sum_{k=1}^{k^*} \sum_{j \in {\cal C}_{k}} \sum_{t=0}^{T} \left( \left( \frac{\widehat{x}_{jt}}{\widetilde{p}_{j}} \left(t + \frac{1}{2} - r_{j}\right) + \frac{\widehat{x}_{jt}}{2} \right) - \left( \frac{x^*_{jt}}{\widetilde{p}_{j}} \left(t + \frac{1}{2} - r_{j}\right) + \frac{x^*_{jt}}{2} \right) \right) \\
    & &  \qquad = \sum_{k=1}^{k^*} \sum_{j \in {\cal C}_{k}} \sum_{t=0}^{T} \frac{t ( \widehat{x}_{jt} - x^*_{jt} )}{\widetilde{p}_{j}} \\
    & &  \qquad = \sum_{k=1}^{k^*} \frac{ 1 }{ 2^k } \sum_{t=0}^{T} t \sum_{j \in {\cal C}_{k}} ( \widehat{x}_{jt} - x^*_{jt} ) \\
    & &  \qquad = \sum_{k=1}^{k^*} \frac{ 1 }{ 2^k } \sum_{t=0}^{T} \left( V_{k,t}(\widehat{x}, \widehat{y}, \widehat{f}) - V_{k,t}( x^*,y^*,f^* ) \right)  \leq  0\\
  \end{MyEqn}%
  The second equality holds since $\sum_{t=0}^{T} \widehat{x}_{jt} =
  \sum_{t=0}^{T} x_{jt}^*$ for every $j$. The last equality holds since
  $\sum_{t=0}^{T} t \sum_{j \in {\cal C}_{k}} \widehat{x}_{jt} =
  \sum_{t=0}^{T} V_{k,t}(\widehat{x}, \widehat{y}, \widehat{f})$ and $\sum_{t=0}^{T}
  t \sum_{j \in {\cal C}_{k}} x_{jt}^* = \sum_{t=0}^{T} V_{k,t}(
  x^*,y^*,f^* )$.

\subsection{Proof of Lemma \ref{lem:stage1cost}: The Stage~I Algorithm} \label{app:stage1proof}
The first part is simple to prove. Observe that the only step where the Stage~I algorithm would schedule a larger job while deleting a smaller job is in Step~4a (or Step~4c). However, even in this case, since it operates within a particular class, the worst case would be scheduling $\Delta$ fraction of a job $j$ in some free time, while deleting $\Delta$ fraction of a job $j'$ whose processing time is half that of $j$. Therefore, it follows that ${\cal P}(x', y', f') \leq 2 {\cal P}(x^*, y^*, f^*)$.

For the proof of the next two parts, we describe an equivalent form of the swapping algorithm which operates on each time slot $[t,t+1)$ one by one, rather than job by job. As the algorithm proceeds, the solution $(x',y',f')$ keeps
getting updated: initially $(x', y', f') = (x^*, y^*, f^*)$.
We introduce a {\em charging scheme}
$\charge(t,j)$ which is initially set to $0$ for every $t \in [0,T]$ and
$j \in \{1, \ldots, n\}$. It is progressively modified over the course
of the algorithm, and helps us bound the increase in LP cost as the jobs
are make integral.

\medskip \noindent {\bf Revisiting the Swapping phase for class-$\bs{k}$.} If the number of
fully scheduled class-$k$ jobs in the current solution is smaller than $\lfloor
\sum_{j \in {\cal C}_k} y_j' \rfloor - 2$, we first \emph{advance} all
class-$k$ jobs as much as possible within the union of class-$k$
time intervals along with the overall free time. That is, we make sure that there
is no free time between the release date of any fractionally scheduled
class-$k$ job and the last time interval in which some fraction of it
is scheduled. Note that this does not cause any increase in cost of the LP solution.

 Now consider some stage of the swapping stage where $j_{1}$ be the first fractionally scheduled class-$k$ job, and let
$j_{2}, \ldots, j_q$ be a prefix of the class-$k$ jobs scheduled after
$j_{1}$ defined thus: $q$ is the minimal index for which $p_{ j_{q+1} }
< p_{ j_1 }$ or for which $j_q$ is the last scheduled class-$k$ job (see \lref[Figure]{flowalgo-fig}).
For any $s < q$, let $\free(j_s, j_{s+1})$ be the overall amount of free
time between the last interval in which $j_s$ is scheduled and the first
interval in which $j_{s+1}$ is scheduled\footnote{If $j_q$ is the last
  scheduled class-$k$ job, we define $\free(j_q, j_{q+1}) =
  \infty$.}. There are two cases to consider:

\begin{figure}[!htbp]
\begin{center}
\includegraphics[scale=0.55]{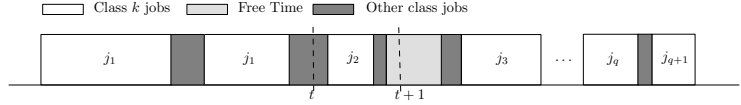}
\end{center}
\vspace{-15pt}
\caption{ {\small A prefix of class-$k$ jobs (schematic illustration).}} \label{flowalgo-fig}
\end{figure}


\noindent\textbf{Case I: $\bs{\sum_{s = 2}^q y_{ j_s }' + \sum_{s
      = 1}^q \free( j_s, j_{s+1} ) / p_{ j_1 } \geq 1 - y'_{j_{1}}}$.}
In this scenario, repeat until $j_1$ becomes fully scheduled:
    \begin{enumerate}
    \item Let $s \in \{ 2, \ldots, q \}$ be the minimal index for which $y_{ j_s }' > 0$. If $y_{ j_2 }' = \cdots = y_{ j_q }' = 0$, let $s = q+1$.

    \item \label{no-free-time} If $\free( j_1, j_s ) = 0$, let $[t,t+1)$ be the first time slot where $j_s$ is scheduled. We now replace a $\Delta_t = \min \{ 1 - y_{ j_1 }', x_{ j_s t }' / p_{ j_s } \}$ fraction of $j_s$ by a $\Delta_t$ fraction of $j_{1}$, possibly creating some free time. We also advance all class-$k$ jobs starting from $j_s$ as much as possible (without violating release dates) within the union of all time intervals in which these jobs are scheduled along with the overall free time. An illustration of this step (before the advancement) is shown in \lref[Figure]{flowalgostep-fig}.

\begin{figure}[!htbp]
\begin{center}
\includegraphics[scale=0.55]{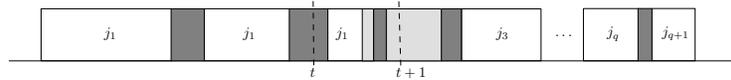}
\end{center}
\vspace{-15pt}
\caption{ {\small Example of Case I. Notice the free time created due to $p_{j_{2}}$ being smaller than $p_{j_{1}}$.}} \label{flowalgostep-fig}
\end{figure}

{\small
\begin{shadebox}
\noindent \underline{Charging}: When a $\Delta_t$ fraction of $j_s$ is replaced by a $\Delta_t$ fraction of $j_{1}$, we say that \emph{each point} of the time interval $(r_{ j_1 }, r_{ j_s })$  pays $\Delta_t p_{ j_1 } / \widetilde{p}_{ j_1 }$ \emph{on behalf} of $j_{1}$. I.e., set $\charge(t',j_1) \leftarrow \charge(t',j_1) + \Delta_t p_{ j_1 } / \widetilde{p}_{ j_1 }$ for all $t' \in (r_{ j_1 }, r_{ j_s })$.
\end{shadebox}

\begin{shadebox}
\noindent \underline{Cost increment}: When we replace $j_s$ with $j_1$, the extra LP cost paid by $j_1$ is
$\frac{ \Delta_t p_{ j_1 } }{ \widetilde{p}_{ j_1 } } \left( t +
\frac{ 1 }{ 2 } - r_{ j_1 } \right) + \frac{ \Delta_t p_{ j_1 }
}{ 2 }$
whereas the cost saved by removing a fraction of $j_s$ is
$\frac{ \Delta_t p_{ j_s } }{ \widetilde{p}_{ j_s } } \left( t + \frac{ 1 }{ 2 } - r_{ j_s } \right) + \frac{ \Delta_t p_{ j_s } }{ 2 }$.
The cost increase is
\[ \ts \left( \frac{ \Delta_t }{ \widetilde{p}_{ j_1 } } \left( t + \frac{ 1 }{ 2 } - r_{ j_s } \right) + \frac{ \Delta_t }{ 2 } \right) ( p_{ j_1 } - p_{ j_s } ) + \frac{ \Delta_t p_{ j_1 } }{ \widetilde{p}_{ j_1 } } ( r_{ j_s } - r_{ j_1 } ) \leq \frac{ \Delta_t p_{ j_1 } }{ \widetilde{p}_{ j_1 } } ( r_{ j_s } - r_{ j_1 } ) \, , \]%
which is exactly the increment in $\int_{0}^{T} \sum_{j \in {\cal C}_k} \charge(t,j) dt$. The inequality above holds because $p_{j_s} \geq p_{j_1}$.

\end{shadebox}
}

    \item \label{some-free-time} Otherwise (i.e., $\free( j_1, j_s ) > 0$) let $I$ be the first free time interval between the last interval in which $j_1$ is scheduled and the first interval in which $j_s$ is scheduled. Also, let $[t,t+1)$ be the first time slot having a non-empty intersection with $I$, and let $j_l$ be the last fractionally scheduled class-$k$ job. Note that $j_l$ cannot be any of the jobs $j_2, \ldots, j_{ q+1 }$, as the number of fully scheduled class-$k$ jobs is at most $\lfloor \sum_{j \in {\cal C}_k} y_j' \rfloor - 2$. Now, we schedule an extra $\Delta_t = \min \{ 1 - y_{ j_1 }', y_{ j_l }', |I \cap [t,t+1) | / p_{j_1} \}$ fraction of $j_1$ in $I \cap [t,t+1)$, while continuously deleting a $\Delta_t$ fraction of $j_l$ from the intervals where this job is scheduled, in reverse order of time.

{\small
\begin{shadebox}
\noindent \underline{Charging}: If $r_{ j_1 } \leq t$, each point of the time interval $(r_{ j_1 }, t)$ pays $\Delta_t p_{ j_1 } / \widetilde{p}_{ j_1 }$ on behalf of $j_{1}$. That is, we set $\charge(t',j_1) \leftarrow \charge(t',j_1) + \Delta_t p_{ j_1 } / \widetilde{p}_{ j_1 }$ for every $t' \in (r_{ j_1 }, t)$.
\end{shadebox}
\begin{shadebox}
\noindent \underline{Cost increment}: The extra cost paid by $j_1$ is
    $\frac{ \Delta_t p_{ j_1 } }{ \widetilde{p}_{ j_1 } } \left( t + \frac{ 1 }{ 2 } - r_{ j_1 } \right) + \frac{ \Delta_t p_{ j_1 } }{ 2 }$,
    while the cost saved by deleting a $\Delta_t$ fraction of $j_l$ is at least $\Delta_t p_{ j_l } / 2$. Since $j_1$ and $j_l$ belong to the same class, the cost increment is at most
    \[ \ts \frac{ \Delta_t p_{ j_1 } }{ \widetilde{p}_{ j_1 } } \left( t - r_{ j_1 } \right) + 3 \cdot \frac{ \Delta_t p_{ j_l } }{ 2 } \ . \]%
    When $r_{ j_1 } \leq t$, the first term is exactly the increment in $\int_{0}^{T} \sum_{j \in {\cal C}_k} \charge(t,j) dt$ as a result of setting $\charge(t',j_1) \leftarrow \charge(t',j_1) + \Delta_t p_{ j_1 } / \widetilde{p}_{ j_1 }$ for every $t' \in (r_{ j_1 }, t)$. In the opposite case, this term is negative, which is why we do not need to modify the charging function. In addition, the term $\Delta_t p_{ j_l } / 2$ lower bounds the contribution of the deleted fraction of $j_l$ towards the quantity $\flow(x^*, y^*, f^*)$.
\end{shadebox}
}
    \end{enumerate}

\noindent\textbf{Case II: $\bs{\sum_{s = 2}^q y_{ j_s }' + \sum_{s = 1}^q \free( j_s, j_{s+1} ) / p_{ j_1 } < 1 - y'_{j_{1}}}$.} In this case, repeat until $j_{ q+1 }$ becomes fully scheduled or until $y_{ j_1 }' = \cdots = y_{ j_q }' = 0$:
    \begin{enumerate}
    \item Let $s \in \{ 1, \ldots, q \}$ be the minimal index for which $y_{ j_s }' > 0$.

  \item Let $[t,t+1)$ be the first time slot where $j_s$ is scheduled. We now replace a $\Delta_t = \min \{ 1 - y_{ j_{q+1} }', x_{ j_s t }' / p_{ j_s } \}$ fraction of $j_s$ by a $\Delta_t$ fraction of $j_{q+1}$, possibly creating some free time.

{\small
\begin{shadebox}
\noindent \underline{Cost increment}: The extra cost paid by $j_{q+1}$ is
      $\frac{ \Delta_t p_{ j_{q+1} } }{ \widetilde{p}_{ j_{q+1} } } \left( t + \frac{ 1 }{ 2 } - r_{ j_{q+1} } \right) + \frac{ \Delta_t p_{ j_{q+1} } }{ 2 }$,
      whereas the cost saved by scheduling a smaller fraction of $j_s$ is
      $\frac{ \Delta_t p_{ j_s } }{ \widetilde{p}_{ j_s } } \left( t +
      \frac{ 1 }{ 2 } - r_{ j_s } \right) + \frac{ \Delta_t p_{ j_s } }{
      2 }$. Since $p_{ s } \geq p_{j_1} > p_{j_{q+1}}$, 
    the cost increment is
      \[ \ts \left( \frac{ \Delta_t }{ \widetilde{p}_{ j_s } } \left( t + \frac{ 1 }{ 2 } - r_{ j_s } \right) + \frac{ \Delta_t }{ 2 } \right) ( p_{ j_{q+1} } - p_{ j_s } ) + \frac{ \Delta_t p_{ j_{q+1} } }{ \widetilde{p}_{ j_s } } ( r_{ j_s } - r_{ j_{q+1} } ) \leq 0 \ , \]%
      and there is no need to modify the charging function.
\end{shadebox}
}
  \end{enumerate}
  Conclude this case by making the following rearrangements:
  \begin{OneLiners}

  \item [1.] Continuously schedule a $y_{ j_{ q+1} }'$ fraction of $j_{ q+1 }$ within the union of class-$k$ intervals and free time, starting at the first interval in which $j_{q+1}$ is currently processed. Even though we may have violated the release date of $j_{ q+1 }$, the earliest time in which any part of this job is processed was advanced by at most $2 \cdot 2^k$ within the union of time intervals where class-$k$ jobs are scheduled and free time intervals, since we initially had $\sum_{s = 2}^q y_{ j_s }' + \sum_{s = 1}^q \free( j_s, j_{s+1} ) / p_{ j_1 } < 1 - y'_{j_{1}}$. This anomaly will be handled in the sequel.

  \item [2.] Following $j_{ q+1 }$, proceed by scheduling a $y_{ j_1 }', \ldots, y_{ j_q }'$ fraction of $j_1, \ldots, j_q$, respectively, as soon as possible (without violating release dates) within the time intervals where class-$k$ jobs are scheduled and free time.
  \end{OneLiners}


We are now ready to prove Lemma \ref{lem:stage1cost}. Noting that the fractional contribution of each class changes only during the course of its corresponding iteration, we may focus our attention on a fixed class-$k$, and bound its fractional cost, $\flow_k(x',y',f') = \sum_{j \in {\cal C}_{k}} F_j'$.

\begin{claim} \label{clm:iter-k-bound}
Just before the shifting phase for class-$k$, we have
\[ \flow_k(x',y',f') \leq 4\cdot \flow_k(x^*,y^*,f^*) + \int_{0}^{T} \sum_{j \in {\cal C}_{k}} \charge(t,j) dt \ . \]%
\end{claim}
\begin{proof}
We can bound the cost increment of each operation in the swapping phase as follows:
\begin{itemize}
\item A single operation in case I, \lref[step]{no-free-time}: As mentioned in the algorithm, the cost increment is upper bounded by the increment in $\int_{0}^{T} \sum_{j \in {\cal C}_{k}} \charge(t,j) dt$.

\item A single operation in case I, \lref[step]{some-free-time}: In these settings, the cost increment can be bounded by the increment in $\int_{0}^{T} \sum_{j \in {\cal C}_{k}} \charge(t,j) dt$ plus thrice whatever the deleted fraction of $j_l$ contributes to $\flow(x^*, y^*, f^*)$. It is important to observe that deleted fractions will not be used later on to bound additional cost increments for subsequent operations for this stage.

\item Operations in case II: Because we are only rearranging jobs within class-$k$ space (and never introduce any free time), arguments similar to the proof in Appendix~\ref{sec:flowalgo-proofs-I} show that no extra cost is incurred in this step.
\end{itemize}%
Therefore, at the completion of the swapping phase, we have
\[ \flow_k(x',y',f') - \flow_k(x^*,y^*,f^*) \leq \int_{0}^{T} \sum_{j \in {\cal C}_{k}} \charge(t,j) dt + 3\cdot \flow_k(x^*,y^*,f^*) \ . \]%
\vspace{-40pt}

\end{proof}

We proceed by establishing a few crucial properties of the charging function.

\begin{claim} \label{free-time-lemma}
Just before the shifting phase, no free time ever pays on behalf of any job. In other words, if $\sum_{j \in {\cal C}_k} \charge(t,j) > 0$ then $t$ cannot be free time.
\end{claim}
\begin{proof}
We prove the above claim by arguing that, whenever an interval is charged, each of its points is currently dedicated to processing some job. It is not difficult to verify that our algorithm preserves this property till the swapping phase terminates. Consider an operation where some interval is charged.
\begin{itemize}
\item In case I, \lref[step]{no-free-time}, suppose that a $\Delta_t$
  fraction of $j_s$ is replaced by a $\Delta_t$ fraction of
  $j_{1}$. Then, the charging scheme increases $\charge(t',j_1)$ by
  $\Delta_t p_{j_1}/\widetilde{p}_{j_1}$ for every $t' \in
  (r_{j_1},r_{j_s})$. However, we are guaranteed not to have free time
  in the interval $(r_{j_1},r_{j_s})$ because of the fact that there is no
  free time between $r_{j_1}$ and the last interval in which $j_1$ is
  scheduled (any such free time is eliminated when we begin the swapping
  phase for class-$k$), and because $\free(j_1,j_s) =0$.

\item In case I, \lref[step]{some-free-time}, the algorithm picks the first time slot (say, $[t_1,t_1+1)$) which has some free time between the last interval in which $j_1$ is scheduled and the first interval in which $j_s$ is scheduled. Suppose that an extra $\Delta_{t_1}$ fraction of $j_1$ is scheduled, increasing $\charge(t',j_1)$ by $\Delta_{t_1} p_{j_1}/\widetilde{p}_{j_1}$ for every $t' \in (r_{j_1},t_1)$. Note that there cannot be free time between the last interval in which $j_1$ is scheduled and $t_1$ (by the way $[t_1,t_1+1)$ was picked), and also between $r_{j_1}$ and the last interval in which $j_1$ is scheduled. Therefore, there is no free time in $(r_{j_1},t_1)$.
\end{itemize}%
\vspace{-20pt}
\end{proof}

\begin{claim} \label{unique-payment-lemma}
Each point in time pays on behalf of at most one job. That is, for every $t \in [0,T]$,
\[ \left| \left\{ j \in {\cal C}_k : \charge(t,j) > 0 \right\} \right| \leq 1 \ . \]%
\end{claim}
\begin{proof}
We prove the above claim by contradiction. For this purpose, suppose there exists a point in time $t^* \in [0,T]$ that pays on behalf of two jobs, say $j$ and $j'$. By the way we charged jobs in the swapping phase, we know that both $j$ and $j'$ must be fully scheduled. Without loss of generality, we assume that $j$ appears before $j'$ in the schedule $(x^*, y^*, f^*)$.\footnote{This assumption implies $r_j \leq r_{ j' }$ because the LP solution is assumed to be non-alternating.} Consider a single operation in which $t^*$ is charged, paying some amount of behalf of $j$.
\begin{itemize}
\item Suppose \lref[step]{no-free-time} is executed in time slot $[t,t+1)$, where $\Delta_t$ fraction of a job $j_s$ is replaced by $\Delta_t$ fraction of $j$ during which $t^*$ pays on behalf of $j$. Since $t^*$ is charged in this operation, it follows that $t^* \in (r_j,r_{j_s})$. In addition, we have $r_{ j_s } \leq r_{j'}$, or otherwise $j'$ must have been fully replaced by $j$ during previous operations, since $j$ is currently replacing $j_s$. Therefore, $t^* < r_{ j_s } \leq r_{ j' }$, implying that $t^*$ cannot be paying on behalf of $j'$, since our charging scheme guarantees that a time point can pay on behalf of a particular job only when it appears after the release date of this job.

\item On the other hand, suppose \lref[step]{some-free-time} is executed in time slot $[t,t+1)$, where a $\Delta_t$ fraction of $j$ is scheduled. Since $t^*$ is charged, it follows that $t^* \in (r_j,t)$. Now, if $t^*$ pays on behalf of $j'$, we must have $r_{ j' } < t^* < t$, meaning that at the moment there is some free time between $r_{j'}$ and the first interval in which $j'$ is scheduled. Such free time would have been eliminated at the beginning of this iteration as a result of advancing class-$k$ jobs.
\end{itemize}
\vspace{-20pt}
\end{proof}

\begin{claim} \label{clm:no-increase-in-p}
$\int_{0}^{T} \sum_{j \in {\cal C}_{k}} \charge(t,j) dt \leq 2 {\cal P}(x^*, y^*, f^*)$.
\end{claim}
\begin{proof}
By \lref[Claim]{unique-payment-lemma}, we know that each point in time pays on behalf of at most one job per class. Also, whenever a point $t$ is charged on behalf of a job $j$, the increment in $\charge(t,j)$ is of the form $\Delta p_{j} / \widetilde{p}_{j}$, where $\Delta$ is the additional fraction of $j$ being scheduled. Therefore, the total amount $t$ can pay on behalf of $j$ is at most $p_j / \widetilde{p}_j \leq 1$. This bound, coupled with the observation that free time never pays on behalf of any job (see \lref[Claim]{free-time-lemma}), proves that
\[ \int_{0}^{T} \sum_{j \in {\cal C}_{k}} \charge(t,j) dt \leq {\cal P}(x',y',f') \leq 2{\cal P}(x^*, y^*, f^*) \ . \]%
The last inequality holds since, throughout stage II, the overall processing time cannot grow by a factor greater than $2$: Whenever we schedule an extra fraction of some class-$k$ job, we also delete an equal fraction from some other class-$k$ job, which in the worst case has half its processing time.
\end{proof}

\begin{claim} \label{clm:iter-k-bound-shifting}
Immediately after the shifting phase, we have
\[ \flow_k(x',y',f') \leq 4\cdot \flow_k(x^*,y^*,f^*) + 6 {\cal P}(x^*, y^*, f^*) \ . \]
\end{claim}
\begin{proof}
By combining \lref[Claims]{clm:iter-k-bound} and \ref{clm:no-increase-in-p}, we can bound the fractional cost of class-$k$ just after the swapping phase by
\begin{MyEqn}
\flow_k(x',y',f') & \leq & 4\cdot \flow_k(x^*,y^*,f^*) + \int_{0}^{T} \sum_{j \in {\cal C}_{k}} \charge(t,j) dt \\
& \leq & 4\cdot \flow_k(x^*,y^*,f^*) + 2 {\cal P}(x^*, y^*, f^*) \ .
\end{MyEqn}%
In addition, arguments nearly identical to those of Garg and Kumar~\cite[Clm.\ 4.3]{GargK07} show that the cost increment due to the shifting phase is at most $2 {\cal P}(x',y',f') \leq 4 {\cal P}(x^*, y^*, f^*)$.
\end{proof}

Part {(ii)} of \lref[Lemma]{lem:stage1cost} is now derived by
summing the inequality stated in \lref[Claim]{clm:iter-k-bound-shifting}
over all classes:
\begin{MyEqn}
\flow(x',y',f') & = & \sum_{k = 1}^{ k^* } \flow_k(x',y',f') \\
& \leq & 4 \sum_{k = 1}^{ k^* } \flow_k(x^*,y^*,f^*) + 6 k^* {\cal P}(x^*, y^*, f^*) \\
& = & 4 \cdot \flow(x^*, y^*, f^*) + 6 k^* {\cal P}(x^*, y^*, f^*) \ .
\end{MyEqn}%

The third part of \lref[Lemma]{lem:stage1cost} follows from the
next lemma.
\begin{lemma} \label{lem:flow_time_cost}
  The sum of flow times of all integrally scheduled jobs is at most $2
  \cdot \flow(x',y',f') + k^* {\cal P}(x',y',f')$.
\end{lemma}

\begin{proof}
  Consider some fully scheduled job $j$. It is easy to verify (see
  \cite{GargK06} for a proof) that the quantity $2F_{j}'$ is at least
  the actual flow time of $j$ minus the amount of time for which $j$ has
  been preempted (which cannot include free time). In addition, our
  algorithm ensures that, at any point in time, at most one class-$k$
  job may be preempted. Hence, by summing over all fully scheduled class
  $k$ jobs, it follows that their sum of flow times is bounded by $2
  \sum_{j \in {\cal C}_{k}} F_{j}' + {\cal P}(x',y',f') = 2 \cdot
  \flow_k(x',y',f') +{\cal P}(x',y',f') $. The desired result is
  obtained by summing over all classes.
\end{proof}

\subsection{Parallel Machines Algorithm}
\label{sec:parallel-machines-algo}

For completeness, we first provide the natural extension of the flow time LP for identical machines.
Then, we present the algorithm in more detail.

\[ \begin{array}{lll}
\mbox{minimize} & { \sum_{j = 1}^{n} f_j} \\
\mbox{subject to} & (1) \quad { \displaystyle f_{j} = \sum_{t=0}^{T} \sum_{i=1}^{m} \left( \frac{x_{ijt}}{\widetilde{p}_{j}} \left(t + \frac{1}{2} - r_{j}\right) + \frac{x_{ijt}}{2} \right)} \qquad & \forall \, j \\
& (2) \quad { p_j y_j = \sum_{t=0}^{T} \sum_{i=1}^{m} x_{ijt}} & \forall \, j \\
& (3) \quad { \sum_{j = 1}^n x_{ijt} \leq 1} & \forall \, t, i \\
& (4) \quad { \sum_{j = 1}^n y_{j} \geq \Pi} \\
& (5) \quad x_{ijt} = 0 & \forall \, i, j, t : t < r_j \\
& (6) \quad x_{ijt} \geq 0, \, 0 \leq y_j \leq 1 & \forall \, i, j, t
\end{array} \]%

\medskip \noindent {\bf The Algorithm:}

\medskip \noindent {\bf Stage I:} Let $(x^*,y^*,f^*)$ be an optimal LP solution.
We first rearrange the jobs following the
procedure given in~\cite{GargK06} to make each job's schedule
non-migratory whilst preserving the fraction to which it has been scheduled.
Let $(\widehat{x},\widehat{y},\widehat{f})$ be the
updated solution. A proof identical to Lemma $3.3$ in \cite{GargK06}
 shows that $\flow(\widehat{x},\widehat{y},\widehat{f})
\leq \flow(x^*,y^*,f^*) + O(\log P) {\cal P}(x^*,y^*,f^*)$ and that ${\cal
  P}(\widehat{x}, \widehat{y}, \widehat{f}) \leq {\cal P}(x^*,y^*,f^*)$.

\medskip \noindent {\bf Stage II:} Let $\flow^{i}(x,y,f)$ of an LP solution $(x,y,f)$ be  $\sum_{j} \sum_{t=0}^{T} \left( \frac{x_{ijt}}{\widetilde{p}_{j}} \left(t + \frac{1}{2} - r_{j}\right) + \frac{x_{ijt}}{2} \right)$, and ${\cal P}^{i}(x,y,f) = \sum_{j} \sum_{t=0}^{T} x_{ijt}$. For each machine, run Stage~I of the single machine algorithm: let $(x',y',f')$ be the (possibly infeasible) solution obtained. From the analysis of the single machine case, we have $\flow^{i}(x',y',f') \leq 2 \cdot \flow^{i}(\widehat{x},\widehat{y},\widehat{f}) + 6(\log P) {\cal P}^{i}(\widehat{x},\widehat{y},\widehat{f})$, and ${\cal P}^{i}(x',y',f') \leq 2{\cal P}^{i}(\widehat{x},\widehat{y},\widehat{f})$. Further, for each $k$, the number of class-$k$ jobs completely scheduled  on any machine $i$ in $(x', y', f')$ is at least the fractional number of class-$k$ jobs scheduled on machine $i$ by $(x^*, y^*, f^*)$ (up to an additive $2$ jobs). 

\medskip \noindent {\bf Stage III:} We now handle the infeasibility of $(x',y',f')$: each class may still have up to $2$ fractionally scheduled jobs on each machine. We set $(\widetilde{x},\widetilde{y},\widetilde{f}) := (x',y',f')$, and make changes to $(\widetilde{x},\widetilde{y},\widetilde{f})$. Like in the single machine case, we \emph{swap} jobs to make them integrally scheduled. For each $k$,

\begin{itemize}
   \item[{\bf IIIa:}] Advance all class-$k$ jobs as much as possible (within time occupied by class-$k$ jobs and free time) such that there is no free time between when a job is released and when it is scheduled in $(\widetilde{x},\widetilde{y},\widetilde{f})$.

    \item[{\bf IIIb:}] Repeat the following until there is at most $1$ fractionally scheduled class-$k$ job in $(\widetilde{x},\widetilde{y},\widetilde{f})$:

         ~~ Let $j_{1}$ be the fractionally scheduled class-$k$ job with largest processing time and $j_{2}$ be the one with smallest processing time. Keep adding $j_{2}$ to the end of the schedule on the machine in which it has currently been scheduled, while deleting an equal fraction from $j_{1}$ until either (i) $j_{2}$ is fully scheduled, or (ii) $j_{1}$ has been completely deleted.

\end{itemize}

\medskip \noindent {\bf Analysis.}
Suppose the algorithm is replacing $j_{1}$ (scheduled on machine $i_{1}$) with $j_{2}$ (scheduled on machine $i_{2}$). Instead of performing the replacement in one shot, we could also do it in a time slot by time slot basis.
Let the last time interval in which $j_{1}$ is scheduled be $[t_{1},t_{1}+1)$, and let $[t_{2}, t_{2}+1)$ be first interval that has free time, after the fractional completion of $j_{2}$. The algorithm deletes a fraction $\Delta = \min(\widetilde{x}_{i_{1}j_{1}t}/p_{j_{1}}, 1 - \widetilde{y}_{j_{2}}, (1 - \sum_{j} \widetilde{x}_{i_{2}jt_{2}})/p_{j_{2}} )$  of $j_{1}$ and schedules $\Delta$ fraction of $j_2$ in the free time in $[t_{2},t_{2}+1)$ on machine $i_{2}$. Intuitively, $\Delta$ is the minimum of the fraction of $j_{1}$ that is scheduled in $[t_{1},t_{1}+1)$, the fraction of $j_{2}$ needed to make it fully scheduled, and the fraction of $j_{2}$ that can be scheduled in the free time in $[t_{2},t_{2}+1)$.

Observe that because $p_{j_2} \leq p_{j_1}$, the additional cost incurred by the modified LP solution is at most
\[
 (t_{2} - r_{j_{2}} + \frac{1}{2}) \Delta + \frac{\Delta p_{j_{2}}}{2} - (t_{1} - r_{j_{1}} + \frac{1}{2}) \Delta + \frac{\Delta p_{j_{1}}}{2}  \leq  (t_{2} - r_{j_2}) \Delta_{t}
\]
Also notice that every point in the interval $(r_{j_2}, t_2)$ is not free time, by the way $t_2$ was chosen.

We then employ a charging scheme where each point $t$ on the time interval $(r_{j_2}, t_{2})$ pays an additional charge of $\Delta$ towards job $j_{2}$. The following properties are then true at the end of this stage:
\begin{OneLiners}
\item [(a)] Each point pays at most $2$ on behalf of jobs belonging to a class on each machine. This is because there can be at most $2$ fractional jobs per class on each machine in the solution $(x', y', f')$ and a point in time pays only for a fractional job which becomes completely scheduled.
\item [(b)] Any point which pays on behalf of a job cannot be ``free time''. This follows because we are guaranteed that there is no free time in any charging interval.
\item [(c)] The total processing time in the LP solution does not increase: ${\cal P}(\widetilde{x},\widetilde{y},\widetilde{f}) \leq {\cal P}(x',y',f')$. This holds because we always replace a fraction of a larger job with an equal fraction of a smaller job.
\end{OneLiners}

Therefore, at the end of this stage, the cost of the updated LP solution $(\widetilde{x},\widetilde{y},\widetilde{f})$ is bounded by
  \[ \flow(\widetilde{x},\widetilde{y},\widetilde{f}) \leq \flow(x',y',f') + 2 (\log P){\cal P}(x',y',f') \]
   and the total processing time by
  \[ {\cal P}(\widetilde{x},\widetilde{y},\widetilde{f}) \leq {\cal P}(x',y',f') \ . \]

 \medskip \noindent {\bf Stage IV:} After Stage III, we might still have at most one fractional job per
  class. To handle this, for each class-$k$, we completely schedule the
  last remaining fractionally scheduled class-$k$ whenever possible on
  the machine in which it has been fractionally scheduled. The analysis
  for this step, and the one for bounding the actual sum of flow times of the integrally scheduled jobs
  is analogous to the one for the single machine case. This proves Theorem   \ref{thm:parallel-flowtime}.



\section{Average Weighted Completion Time} \label{sec:app-wjcj}

\subsection{Weighted Completion Time with $\bs{K}$ Profit Constraints} \label{sec:proofs-wjcj-ext}

We now consider an extension of the problem studied in \lref[Section]{sec:wsoct}: one in which there are $K$ different profit requirements of the form $\sum_{j} {\pi}^{k}_{j} y_{j} \geq \Pi^k$ for $1 \leq k \leq K$. We highlight the changes to be made to our algorithm and then present its analysis.

\medskip \noindent {\bf Necessary modifications:}
\begin{OneLiners}
\item [(a)] KC constraints for each profit requirement are written down in the LP. Analogous to the single profit requirement case, we define $\pi^{k,{\cal A}}_{j} = \min \{ \pi^{k}_{j}, \Pi^k - \Pi^{k}({\cal A}) \}$ for each subset of jobs ${\cal A}$ and $1 \leq k \leq K$.

\item [(b)] We set ${\cal A}^* = \{j : \widehat{y}_{j} \geq 1/\beta_K \}$. That is, instead of rounding up each $\widehat{y}_{j}$ by a factor of $2$, we round these variables by a factor of $\beta_K$, a parameter whose value will be determined later.

\item [(c)] For jobs in ${\cal A}^*$, we mark each machine/time pair $\tau_j = (i,t)$ with probability $\widehat{x}_{ ijt } / ( p_{ ij } \widehat{y}_j )$. For jobs not in ${\cal A}^*$, we mark each machine/time pair $\tau_j = (i,t)$ with probability $\beta_K \widehat{x}_{ ijt } /  p_{ ij }$. Essentially, we pick jobs in ${\cal A}^*$ with probability $1$, and every other job with probability $\beta_K \widehat{y}_{j}$.
\end{OneLiners}

\medskip \noindent {\bf Analysis.} The proof that the expected cost is
within a factor of $O(\beta_K)$ of optimal is nearly identical to that of the
single profit requirement case. Therefore, we would like to fix $\beta_K$ such that
all profit constraints are simultaneously satisfied with constant`
probability. To this end, consider one such profit requirement $\Pi^k$.
We upper bound the probability that the collection of jobs picked does not satisfy
this requirement. Consider the knapsack cover inequality for ${\cal A}^*$ with respect to requirement $k$, stating that $\sum_{j \notin {\cal A}^*}
  \pi^{k,{{\cal A}^*}}_{j} \widehat{y}_{j} \geq \Pi^k - \Pi^{k}( {\cal A}^* )$.
The total profit collected from
  jobs not in ${\cal A}^*$ can be lower bounded by $Z = \sum_{j \notin {\cal A}^*}
  \pi^{k,{{\cal A}^*}}_{j} Z_{j}$; here, each $Z_j$ is a random variable
  indicating whether job $j$ is picked. To provide an upper bound on the probability that $Z$ falls below $\Pi^k
  - \Pi^{k}( {\cal A}^*)$, note that
{\small
\[ \ts \ex{Z} = \ex{\sum_{j \notin {\cal A}^*} \pi^{k,{{\cal A}^*}}_{j} Z_{j}} = \beta_K \sum_{j \notin {\cal A}^*} \pi^{k,{{\cal A}^*}}_{j} \widehat{y}_{j} \geq \beta_K(\Pi^k - \Pi^{k}( {\cal A}^*)) \ . \]%
}
Consequently, let us define $\alpha_j = \pi^{k,{{\cal A}^*}}_{j} / (\Pi^k - \Pi^{k}( {\cal A}^*))$. Since our algorithm \emph{independently} picks each job not in ${\cal A}^*$ with probability $\beta_K \widehat{y}_{j}$, we have
{\small
\begin{MyEqn}
\pr{ Z \leq \Pi^k - \Pi^{k}( {\cal A}^*) } & = & \ts \pr{ \sum_{j \notin {\cal A}^*} \frac{ \pi^{k,{{\cal A}^*}}_{j} }{ \Pi^k - \Pi^{k}( {\cal A}^*) } Z_{j} \leq 1 } \leq \pr{ \sum_{j \notin {\cal A}^*} \alpha_j Z_{j} \leq \frac{ \ex{ \sum_{j \notin {\cal A}^*} \alpha_j Z_{j} } }{ \beta_K } } \\
& \leq & \ts \exp \left( -\frac{{(\beta_K -1 )}^{2}}{2\beta_K} \right)  \ ,
\end{MyEqn}%
}
where the first and third inequalities hold since $\expar{ \sum_{j
    \notin {\cal A}^*} \alpha_j Z_{j} } \geq \beta_K$, and the second
inequality follows from the Chernoff-type bound on the lower tail of the
sum of independent $[0,1]$ r.v.s (see, e.g., \cite[Thm.\ 3.5]{chernoff}).

We then fix $\beta_K$ such that $\exp( - (\beta_K -1)^2/2\beta_K)$ is at most $1/10K$ (it
suffices for $\beta_K$ to be $O(\log K)$ for this to hold). Consequently, by the union bound, the probability that \emph{some} profit constraint will not be satisfied is at most $1/10$. It follows that our randomized algorithm computes a schedule whose
expected cost is $O(\beta_K) \Opt$, and all the profit
constraints are met with $\Omega(1)$ probability.




\end{document}